\documentclass[twocolumn,aps,nobalancelastpage,10pt,superscriptaddress]{revtex4-2}

\makeatletter
\renewcommand{\paragraph}{%
  \@startsection{paragraph}{4}{0pt}%
  {3.25ex plus 1ex minus .2ex}%
  {1.5ex plus .2ex}%
  {\normalfont\normalsize\bfseries}}
\makeatother

\usepackage[colorlinks=true, allcolors=blue]{hyperref}
\usepackage{graphicx}
\usepackage{amsmath}
\usepackage{amsfonts}
\usepackage{amssymb}
\usepackage{mathtools}
\usepackage{braket}
\usepackage{dsfont}
\usepackage{quantikz}
\usepackage{subcaption}

\usepackage{amsthm}

\newtheorem{theorem}{Theorem}[section]

\newtheorem{lemma}[theorem]{Lemma}

\DeclareMathOperator*{\argmin}{arg\,min}

\newlength{\rxw}
\newlength{\ucw}
\newlength{\sgw}
\colorlet{errcol}{red!60!black}

\newcommand{\rxlab}[1]{\makebox[\rxw][c]{\ensuremath{#1}}} 
\newcommand{\uclab}[1]{\makebox[\ucw][c]{\ensuremath{#1}}}   
\newcommand{\sglab}[1]{\makebox[\sgw][c]{\ensuremath{#1}}}   
\begin{document}

\onecolumngrid

\title{Parameter-Dependent Noise Resilience in the Quantum Approximate Optimization Algorithm}

\author{Jorja Kirk}
\affiliation{School of Mathematics, University of Bristol}
\affiliation{Phasecraft Ltd.}

\author{Ashley Montanaro}
\affiliation{Phasecraft Ltd.}
\affiliation{School of Mathematics, University of Bristol}

\date{\today}

\begin{abstract}
The Quantum Approximate Optimization Algorithm (QAOA) is a variational quantum algorithm used for solving combinatorial optimization problems, which is thought to be a promising application of near-term quantum hardware. However, high rates of noise in these devices can obscure meaningful results. Therefore, it is crucial to understand and suppress the effects of noise on the algorithm's output. Here, we find that the stability of the algorithm in the presence of noise is dependent on its variational parameters. Using a Pauli noise model, we derive an upper bound on the total variation distance between the noisy and noiseless circuit output as a function of these parameters. Results from noisy simulations support this bound and give insight into the best methods to use for parameter generation in the presence of noise. We also introduce new methods of optimization using error penalties and truncated search based on this bound to manufacture error resilience. We extend this analysis to the quantum adiabatic algorithm (QAA) by taking the continuous limit of the QAOA. 
\end{abstract}

\maketitle

\section{Introduction}

Optimization problems are thought to be a significant application of quantum computers, with possible use across fields such as finance, power systems, and transportation logistics~\cite{optop}. Near-term quantum hardware is currently under rapid development and these devices represent an important step toward practical quantum computation~\cite{NISQ}. Limited qubit numbers and relatively high levels of noise, however, make it difficult to realize quantum advantage. Variational quantum algorithms (VQAs) offer a promising approach to achieve this~\cite{vqas,vqe,vqes}. VQAs combine classical and quantum computation, where a parametrized quantum circuit is optimized using classical feedback loops. This structure allows for shallower circuits that are more tolerant to noise, making them especially well-suited for current hardware capabilities~\cite{VQA}. As a result, VQAs have become a central focus in the search for near-term applications of quantum computing.

The Quantum Approximate Optimization Algorithm (QAOA)~\cite{farhi} is a variational algorithm designed to tackle combinatorial optimization problems. Inspired by ideas from quantum annealing and the Quantum Adiabatic Algorithm (QAA)~\cite{qaa}, which evolves a system from the ground state of an initial Hamiltonian to the ground state of a cost Hamiltonian encoding the solution, QAOA can be viewed as a trotterized version of this adiabatic evolution~\cite{farhi}. QAOA operates by alternating between applying problem-specific and mixing Hamiltonians to a quantum state, with the goal of preparing a state that approximates the ground state and hence the optimal solution. 

However, training QAOA on real devices remains a substantial challenge. Noisy hardware can obscure meaningful gradients and make optimization difficult, and the presence of barren plateaus can hinder effective training, especially as the algorithm scales~\cite{BP}. Despite this, QAOA continues to be a candidate for demonstrating practical quantum utility in the NISQ era~\cite{qaoagood}. 

In this work we explore how the choice of variational parameters impacts the resilience of the QAOA circuit to errors. This is a natural question, and indeed  Ref.~\cite{gonzalez2022error} explicitly asked whether specific choices of variational parameters could prevent the proliferation of noise in the resulting circuit. In fact, it is easy to see numerically that the choice of parameters does matter in terms of performance under errors. Fig.~\ref{fig:star} shows the energy produced with and without noise for many sets of parameters, for a simple noise model where we insert one random X or Z error in the circuit. We see that the best parameters in the noise free case are not necessarily the best in the noisy case and there is little correlation between the two. This raises the question of how we should choose parameters when evaluating QAOA on real noisy devices. 

\begin{figure*}[t]
    \centering
    \includegraphics[width=\linewidth]{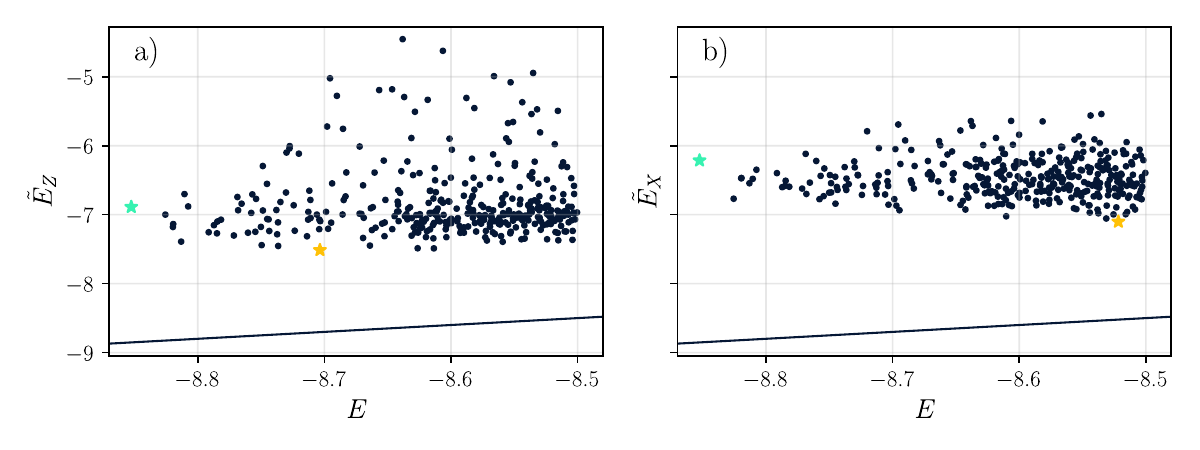}
    \caption{Energies obtained by QAOA circuits experiencing one error, compared with exact energies. Each point represents one set of parameters used in 15 layers of QAOA to solve Max-Cut on a line graph with 10 nodes. Each run was done in simulation to return energy $E$ and then also with simulated errors in the circuit. The energy of the noisy circuit $\tilde{E}_Z$ or $\tilde{E}_X$ was obtained by averaging over the resulting energy with many different error locations of single Pauli $Z$s or $X$s in a) and b) respectively. The highlighted points show the parameters that perform best without errors (teal) and with errors (yellow). The optimal energy for this graph is $E=-9$ and the solid line indicates where the energy with errors equals the noiseless energy.}\label{fig:star}
\end{figure*}

Motivated by this, we undertake a theoretical and numerical study of how the performance of QAOA under errors depends on the choice of parameters. First, we upper bound the total variation distance between distributions gained from circuits with and without errors based on the variational parameters and degree of the graph. 

Averaging over this we obtain an expected upper bound on the total variation distance between the noisy and noiseless distributions that we can use to qualitatively inspect the relative resilience of different sets of parameters. We find smaller valued and smoother varying parameters are more stable in the presence of noise than larger valued, more erratic parameters.

We find that for certain levels of noise taking precomputed parameters, without optimization, is often the most stable approach. While not always performing best in the exact simulated case, these parameters can often provide the best results in the noisy case without the burden of optimization. We construct methods for manufacturing error resilience via additional penalties during optimization and truncating the parameter search space to force smaller valued parameters.

Given the structure of the QAOA parameters that minimize the bound while performing well, we thought it natural to extend this notion to the QAA to find if this gives more insight into suitable QAOA parameter schedules. Taking the infinite time limit of QAOA we recover the QAA and find that our metric of error resilience is lower bounded when using the standard QAA with a single schedule controlling the evolution. We attempt to improve this lower bound by investigating a two schedule approach but find that measures made to attempt to decrease this lower bound increase the runtime of the algorithm such that the bound cannot be improved.

\subsection{Previous Work}

There is evidence to show that QAOA can achieve superior results compared to classical optimization algorithms~\cite{qaoaevidence,qaoaevidence2,qaoaevidence3}. QAOA can also provide a simple but useful model for other quantum algorithms, in particular for quantum simulation~\cite{qaoaforsim}. 

Finding optimal parameters for QAOA is NP-hard~\cite{hardparams}. Methods for the classical optimization of QAOA parameters have been widely studied since noise and time limitations require as few calls to the quantum computer as possible. These include gradient-based and gradient-free methods~\cite{gradient} as well as reinforcement learning~\cite{reinfo}. Similarly, methods are being developed to minimize the need for optimization on hardware irrespective of the classical optimizer, such as warm-starting with accurate guesses for the optimal parameters~\cite{warmstarts,egger2021warm} or parameter prediction strategies which completely bypass the need for optimization~\cite{leo,ieva}.

To utilize the potential of QAOA, it is necessary to ensure the algorithm is robust to the effects of noise. As noted above, VQAs in general tend to suppress noise naturally, as the classical optimization adapts the parameters accordingly~\cite{vqas,fontana2021evaluating}, an effect that has been shown to hold experimentally~\cite{colless2018computation}. Barren plateaus induced by noise in this setting suppress cost gradients exponentially~\cite{BP,BarrenPlateaus}. More recently, work on convergence guarantees under noise has been done for the Variational Quantum Eigensolver (VQE), bounding the shift in the optimal parameters under both coherent and incoherent noise~\cite{resilience_VQEs}. 

For QAOA, work has been done to characterize the effect of noise both numerically and analytically~\cite{xue2021effects, marshall2020characterizing}. This includes proving scaling of output-state fidelity, cost function value, and cost gradient with noise rate, number of qubits and circuit depth yielding an explicit noise-depth trade-off. 

Other results give limits to the possible advantage gained from optimization techniques such as QAOA, including a threshold error rate derived by modelling single qubit error propagation, which is required for the possibility of quantum advantage for VQAs for solving classical optimization problems~\cite{gonzalez2022error}. Similar results were found via relative-entropy inequalities~\cite{stilck2021limitations}. 

\subsection{Organization}

We begin Section \ref{section:preliminaries} by introducing necessary preliminaries including the QAOA, the Max-Cut problem and the QAA and then go on to introduce the noise model we will use throughout. In section \ref{section:theorem} we provide our main result, upper bounding the distance between noisy and noiseless circuit outputs in Theorem \ref{theorem} and provide some numerical results to support this. In section \ref{section:choosing} we investigate the effectiveness of parameters in the presence of noise and introduce new methods of noise-informed parameter selection. Then in section \ref{section:adiabatic} we extend these notions to the QAA. 

\section{Preliminaries}\label{section:preliminaries}

\subsection{QAOA}

Standard QAOA is a VQA with $p$ layers. Each layer consists of a cost and a mixer term with variational parameters $(\gamma_i,\beta_i)$ respectively. Overall, the circuit is parametrized by $2p$ parameters $(\gamma,\beta)$ such that $\gamma = (\gamma_1,\gamma_2,\ldots\gamma_p)$ and $\beta = (\beta_1,\beta_2,\ldots\beta_p)$. 

The goal of QAOA is to find the bitstring(s) $\textbf{z}\in\{0,1\}^n$ that minimize(s) the classical cost function $C(\textbf{z})$ of the optimization problem. The cost Hamiltonian encodes this as
\begin{align}
    H_C  = \sum_{z\in\{0,1\}^n}C(\textbf{z})\ket{\textbf{z}}\bra{\textbf{z}}.
\end{align}
The mixer Hamiltonian, $H_M$, which does not commute with the cost Hamiltonian, is needed in order to explore the solution space.

The full circuit is then
\begin{align}
    U(\gamma,\beta) = U_p\ldots U_1 \quad\text{where}\quad U_k = e^{i\beta_kH_M}e^{i\gamma_kH_C},
\end{align}
applied to an initial state $\ket{\psi_0}$ to give 
\begin{align}
    \ket{\psi(\gamma,\beta)} = U(\gamma,\beta)\ket{\psi_0}.
\end{align}
The initial state is, in general, chosen to be an eigenstate of the mixer, for example, for the standard QAOA mixer $H_M = \sum_{i=1}^nX_i$ we choose $\ket{\psi_0}=\ket{+^n}$. 

This state is then measured and the output is fed into a classical optimizer. To optimize the angles, we construct an objective function as the expectation value of the cost Hamiltonian with respect to the final state
\begin{align}
    F(\gamma,\beta) = \bra{\psi(\gamma,\beta)}H_C\ket{\psi(\gamma,\beta)}.
\end{align}
The classical optimizer is used to search for the optimal parameters $(\gamma^*,\beta^*)$ such that 
\begin{align}
    (\gamma^*,\beta^*) = \argmin _{\gamma,\beta} F(\gamma,\beta).
\end{align}

There are many variants of the QAOA that modify the standard structure~\cite{variants}. This can be via parameter structure, such as multi-angle QAOA~\cite{herrman2022multi} and FALQON~\cite{magann2022feedback}; via the choice of mixer and/or initial state, such as ADAPT-QAOA~\cite{zhu2022adaptive} and warm-started QAOA~\cite{egger2021warm} or via changes to the problem structure, such as RQAOA~\cite{bravyi2020obstacles}. In this paper, we will focus on the standard form of QAOA though there is scope to explore this analysis beyond this.

\subsubsection{Max-Cut}
In this work, we will focus on applying QAOA to the Max-Cut problem. Max-Cut is a fundamental graph optimization problem that is widely used as an application of QAOA \cite{farhi,leo,crooks2018performance,harrigan2021quantum,wang2018quantum}. It involves partitioning a set of graph nodes into two groups via a cut that maximizes the sum of the weights of the edges crossed. For a graph $G(V,E)$ where $V$ is the set of nodes with $|V| = n$ and $E$ is the set of edges, the goal is then to find
\begin{align}
    C^*_{\text{cut}}=\max_{\textbf{z}\in \{\pm1\}^n}C_{\text{cut}}(\mathbf{z}) = \max_{\textbf{z}\in \{\pm1\}^n}\frac{1}{2}\sum_{(i,j)\in E} w_{ij}\left(1 - z_i z_j\right),
\end{align}
which is equivalent to minimizing the Ising energy
\begin{align}\label{maxcutcost}
C^*=\min_{\textbf{z}\in \{\pm1\}^n}C(\textbf{z})=\min_{\textbf{z}\in \{\pm1\}^n}\sum_{(i,j)\in E}w_{ij}z_iz_j,
\end{align}
where $w_{ij}\in \mathbb{R}$ is the edge weight between nodes $i$ and $j$. Throughout, we will focus on the unweighted case where $w_{ij} =1$ for all $(i,j)$.

Finding the optimal solution is an NP-hard problem~\cite{nphard,korte2011combinatorial}. The best guarantee that has been proven for efficient algorithms solving Max-Cut is from the Goemans–Williamson algorithm~\cite{GW,GW1}, which delivers a solution with approximation ratio 0.878, and under the unique games conjecture this is the optimal bound~\cite{ugc}.

Specifically for the Max-Cut problem, the cost function defined in \eqref{maxcutcost} can be mapped to a Hamiltonian whose ground state encodes the solution to this problem
\begin{align}
    H_C = \sum_{(i,j)\in E}w_{ij}Z_iZ_j,
\end{align}
and the mixer Hamiltonian for this problem is typically
\begin{align}
    H_M = \sum_{i=1}^nX_i.
\end{align}
Some of the work done here would more than likely be generalizable to other Ising type problems, but it is future work to verify this.

\subsubsection{Adiabatic evolution}

QAOA in the infinite depth limit can implement the QAA, and QAA can provide useful intuition for QAOA’s performance. In general, the QAA works as follows: consider a quantum system that evolves under the time-dependent Hamiltonian $H(t)$ via the Schr\"odinger equation
\begin{align}\label{schro}
    i\frac{d\ket{\psi(t)}}{dt}=H(t)\ket{\psi(t)}.
\end{align} 
Let $\ket{E_k;t}$ denote the instantaneous  eigenstates of the Hamiltonian, $H(t)\ket{E_k;t} = E_k(t)\ket{E_k;t}$, where $E_k(t)$ are the corresponding eigenvalues and $k$ labels the eigenstates, with $k=0$ the ground state.

The quantum adiabatic theorem states that if the Hamiltonian varies slowly enough in $t$, the state of the system will stay close to the instantaneous ground state of the Hamiltonian at each time $t$, provided the gap $g(t) = E_1(t) - E_0(t)$ is non-zero for all $t$ \cite{born1928beweis,kato1950adiabatic}.

In the form used by Roland and Cerf \cite{roland2002quantum} the adiabatic theorem states, assuming a non-degenerate ground state, and a Hamiltonian free of fast driven oscillations \cite{amin2009consistency},
if we prepare the system at time $t=0$ in its ground state $\ket{E_0;0}$ and let it evolve under the Hamiltonian $H(t)$ for time $T$ to $\ket{\psi(T)}$, then
\begin{align}\label{adiabatic-condition}
    |\braket{E_0;T|\psi(T)}|^2 \geq 1-\eta^2,
\end{align}
provided that
\begin{align}
    \frac{\max_{0\leq t\leq T}\bigg|\bra{E_1;t}\frac{dH}{dt}\ket{E_0;t}\bigg|}{\min_{0\leq t\leq T}g(t)^2}\leq \eta.
\end{align}
for $\eta\ll1$. This is the global condition for adiabaticity. Asserting the same criterion on each infinitesimal time interval gives the local condition
\begin{align}
    \frac{\big|\bra{E_1;t}\frac{dH}{dt}\ket{E_0;t}\big|}{g(t)^2}\leq \eta, \label{localadiabatic}
\end{align}
which is required to hold for all $t$ \cite{roland2002quantum}. 

In general the Hamiltonian takes the form 
\begin{align}\label{sym}
    H(t) = s(t) H_{\text{final}} + (1-s(t)) H_{\text{initial}},
\end{align}
for some schedule $s(t)$ where $s(0)=0$ and $s(T)=1$ and $H_{\text{initial}}$ is a Hamiltonian of which the initial state of our system is a ground state and $H_{\text{final}}$ is some Hamiltonian whose ground state encodes the solution to the problem we are trying to solve. In relation to QAOA, $H_{\text{initial}}=-H_M$ and $H_{\text{final}}=H_C$.

To satisfy \eqref{localadiabatic}, it is sufficient to choose s(t) such that
\begin{align}
    \left|\frac{ds}{dt}\right| \leq \eta\; \frac{g^2(t)}{\left|\bra{E_1;t}\frac{dH}{ds}\ket{E_0;t}\right|}.
\end{align}
$s(t)$ is chosen to satisfy these bounds whilst minimizing the total time the algorithm takes.

This form of the adiabatic theorem comes from the first-order estimation of the transition amplitude rather than a rigorous theorem. Roland and Cerf cite Bransden and Joachain's \textit{Quantum Mechanics} \cite{bransden2000quantum} (and an often cited additional source is A. Messiah's \textit{Quantum Mechanics} \cite{messiah1962quantum}), which use leading-order estimates in the non-adiabatic coupling to give the $\bra{E_1}\dot{H}\ket{E_0}/g^{2}$ criterion.

Conditions of this form have been shown to not be sufficient in general \cite{tong2005quantitative}. Rigorous versions carry further terms involving the second derivative of $H$ and higher powers of $\tfrac{1}{g}$. In particular Theorem 3 from \cite{jansen2007bounds} bounds the distance between the evolved state and the instantaneous ground subspace by terms of this form.

However, we use the first-order approximation here because the violations of this condition are argued to arise from resonant transitions between energy levels, so that for Hamiltonians with no fast driven oscillations the original condition holds \cite{amin2009consistency}. The schedules we use are smooth and monotone and for our examples $H(t)$ has no driving frequency, hence this approximation suffices.

\subsection{Modelling and Characterizing Noise}

Modelling noise exactly can be very challenging~\cite{lidarandbrun} and different quantum devices experience different kinds of noise to varying degrees. In order to approximate and therefore combat this noise, there are various channels that can be chosen to model it. For example, depolarizing noise corresponds to bit and/or phase flips occurring throughout the circuit and obscures meaningful results by driving the state toward the maximally mixed state~\cite{nielsen2010quantum}.

It has been shown that by applying techniques such as Pauli twirling or randomized compiling, arbitrary noise channels can be transformed into effective Pauli channels~\cite{twirling,twirling2,twirling3}. 
\begin{align}
    \mathcal{P}(\rho) = (1-\epsilon)I\rho I + \epsilon_X X\rho X + \epsilon_Y Y\rho Y +\epsilon_Z Z \rho Z,\label{Pauli}
\end{align}
where $\epsilon = \sum_{K\in\{X,Y,Z\}}\epsilon_K$. In particular, when the two-qubit gates are Clifford gates, these methods ensure that the logical circuit remains unchanged while the noise is simplified to a stochastic Pauli form. This process involves averaging the error channel over a set of Pauli operators, effectively symmetrizing the noise so that it behaves like a probabilistic mixture of Pauli errors, which are easier to analyse and simulate. From this point, we assume noise will take the form of Pauli channels given in \eqref{Pauli}. In simulation, this is done by the random injection of single Paulis on each qubit after each layer, sampled according to the probabilities $(1-\epsilon,\epsilon_X,\epsilon_Y,\epsilon_Z)$. 

\subsubsection{Propagating Errors}

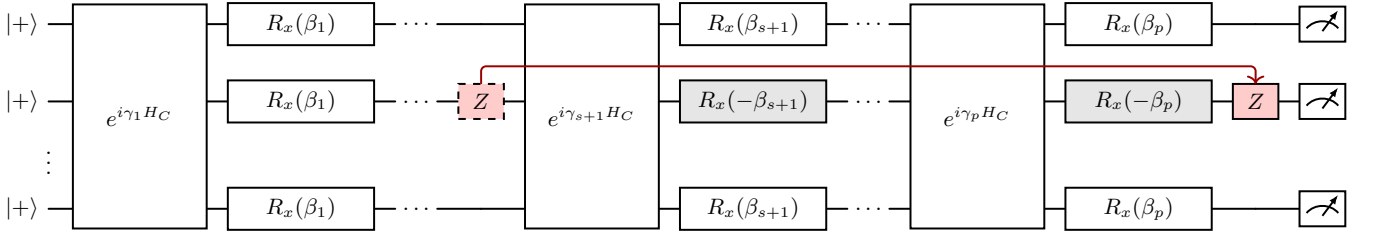
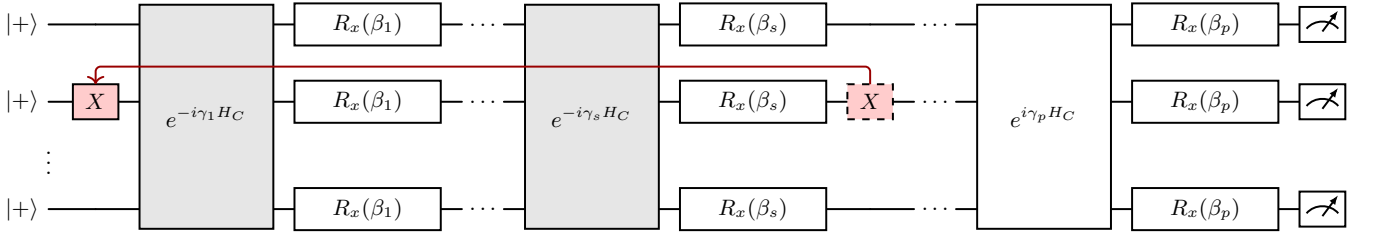
\begin{figure*}[t]
  \centering
  \begin{subfigure}{\textwidth}
    \centering
    \resizebox{\textwidth}{!}{%
    \begin{quantikz}[row sep=0.2cm, column sep=0.3cm,
      execute at end picture={%
        \coordinate (laneA) at ([yshift=2mm]\tikzcdmatrixname-2-7.north);%
        \draw[->,thick,errcol,rounded corners=2pt]
          (\tikzcdmatrixname-2-5.north)
          -- (\tikzcdmatrixname-2-5.north |- laneA)
          -- (\tikzcdmatrixname-2-11.north |- laneA)
          -- (\tikzcdmatrixname-2-11.north);%
      }]
      \lstick{$\ket{+}$}
        & \gate[4]{\uclab{e^{i\gamma_1 H_C}}}
        & \gate{\rxlab{R_x(\beta_1)}}
        & \ \ldots\
        &
        & \gate[4]{\uclab{e^{i\gamma_{s+1} H_C}}}
        & \gate{\rxlab{R_x(\beta_{s+1})}}
        & \ \ldots\
        & \gate[4]{\uclab{e^{i\gamma_p H_C}}}
        & \gate{\rxlab{R_x(\beta_p)}}
        &
        & \meter{}
      \\[0.3cm]
      \lstick{$\ket{+}$}
        &
        & \gate{\rxlab{R_x(\beta_1)}}
        & \ \ldots\
        & \gate[style={dashed,fill=red!20}]{\sglab{Z}}
        &
        & \gate[style={fill=gray!20,thick}]{\rxlab{R_x(-\beta_{s+1})}}
        & \ \ldots\
        &
        & \gate[style={fill=gray!20,thick}]{\rxlab{R_x(-\beta_p)}}
        & \gate[style={fill=red!20}]{\sglab{Z}}
        & \meter{}
      \\
      \vdots
      \\
      \lstick{$\ket{+}$}
        &
        & \gate{\rxlab{R_x(\beta_1)}}
        & \ \ldots\
        &
        &
        & \gate{\rxlab{R_x(\beta_{s+1})}}
        & \ \ldots\
        &
        & \gate{\rxlab{R_x(\beta_p)}}
        &
        & \meter{}
    \end{quantikz}}
    \caption{$Z$ error propagated to the end of the circuit changes the signs of parameters on the mixer terms.}
    \label{fig:qaoa-a}
  \end{subfigure}
\\[0.7cm]
  \begin{subfigure}{\textwidth}
    \centering
    \resizebox{\textwidth}{!}{%
    \begin{quantikz}[row sep=0.2cm, column sep=0.3cm,
      execute at end picture={%
        \coordinate (laneB) at ([yshift=2mm]\tikzcdmatrixname-2-7.north);%
        \draw[->,thick,errcol,rounded corners=2pt]
          (\tikzcdmatrixname-2-8.north)
          -- (\tikzcdmatrixname-2-8.north |- laneB)
          -- (\tikzcdmatrixname-2-2.north |- laneB)
          -- (\tikzcdmatrixname-2-2.north);%
      }]
      \lstick{$\ket{+}$}
        &
        & \gate[4,style={fill=gray!20,thick}]{\uclab{e^{-i\gamma_1 H_C}}}
        & \gate{\rxlab{R_x(\beta_1)}}
        & \ \ldots\
        & \gate[4,style={fill=gray!20,thick}]{\uclab{e^{-i\gamma_{s} H_C}}}
        & \gate{\rxlab{R_x(\beta_{s})}}
        &
        & \ \ldots\
        & \gate[4]{\uclab{e^{i\gamma_p H_C}}}
        & \gate{\rxlab{R_x(\beta_p)}}
        & \meter{}
      \\[0.3cm]
      \lstick{$\ket{+}$}
        & \gate[style={fill=red!20}]{\sglab{X}}
        &
        & \gate{\rxlab{R_x(\beta_1)}}
        & \ \ldots\
        &
        & \gate{\rxlab{R_x(\beta_{s})}}
        & \gate[style={dashed,fill=red!20}]{\sglab{X}}
        & \ \ldots\
        &
        & \gate{\rxlab{R_x(\beta_p)}}
        & \meter{}
      \\
      \vdots
      \\
      \lstick{$\ket{+}$}
        &
        &
        & \gate{\rxlab{R_x(\beta_1)}}
        & \ \ldots\
        &
        & \gate{\rxlab{R_x(\beta_{s})}}
        &
        & \ \ldots\
        &
        & \gate{\rxlab{R_x(\beta_p)}}
        & \meter{}
    \end{quantikz}}
    \caption{$X$ error propagated to the beginning of the circuit changes the signs of parameters on the cost terms on nodes adjacent to the error.}
    \label{fig:qaoa-b}
  \end{subfigure}
  \caption{QAOA circuit diagram with a single Pauli (a) $Z$ error  or (b) $X$ error inserted after layer $s$. This shows how the errors can be propagated to the end or beginning of the circuit, flipping the signs of the parameters on the non-commuting terms.}
  \label{fig:qaoa-both}
\end{figure*}

We start by considering the simple setting where at most one Pauli error occurs within the circuit between layers. We denote QAOA with a single Pauli error $K\in\{X,Y,Z\}$ on qubit $j$ after layer $s$ as 
\begin{align}
    \tilde{U}^K_{j,s}(\gamma,\beta) = U_p \ldots U_{s+1} K_j U_s \ldots U_1.
\end{align}
We first consider the subset $K \in\{X,Z\}$. Commuting the error through one layer of QAOA flips the sign of the angle, $\beta_s$ or $\gamma_s$, of each term that anticommutes with $K$, while leaving the commuting terms unchanged
\begin{align}
    e^{i\beta X}Z &= Z e^{-i\beta X},\\
    e^{i\gamma Z\otimes Z}(X\otimes I) &= (X\otimes I)e^{-i\gamma Z\otimes Z}.
\end{align}

We write $[A,B]_+ \equiv AB+BA$ for the anti-commutator, and define $\chi_{K,P} \equiv \mathds{1}\{[K,P]_+=0\}$ for $P\in\{X,Z\}$. Propagating the Pauli through one QAOA layer gives
\begin{align}\label{big}
    U_s K_j
    &= K_j \hat{U}_s \notag \\
    &= K_j\,
        e^{\,i\beta_s (-1)^{\chi_{K,X}} X_j}\,
        e^{\,i\beta_s \sum_{i\neq j} X_i}  \\
    &\quad\times\,
        e^{\,i\gamma_s (-1)^{\chi_{K,Z}} \sum_{i\in\mathcal{N}(j)} Z_i Z_j}\,
        e^{\,i\gamma_s \sum_{(i,k)\in E'} Z_i Z_k},\notag
\end{align}
where $\mathcal{N}(j)$ denotes the set of neighbouring nodes of $j$ and $E' = \{(i,k)\in E; i,k \notin \mathcal{N}(j)\}$ the set of edges not connected to node $j$. 

It is sufficient to consider errors only between layers. An error occurring within a layer, i.e. between the cost and mixer terms, can be propagated either forward or backward in the circuit, through terms it commutes with, such that it is between layers. This is then equivalent to the error having occurred at that boundary.

Commuting the error through to the beginning or end of the circuit gives successive phase flips. If we have a Pauli $Z$ error we can propagate it to the end where we measure in the computational basis where the resulting $Z$ has no effect on the output distribution
\begin{align}
    \tilde{U}^Z_{j,s} &= U_p \ldots U_{s+1} Z_j U_s \ldots U_1\\
    & = Z_j \hat{U}_p \ldots \hat{U}_{s+1} U_s \ldots U_1\label{dragZ}.
\end{align}
If we propagate the $X$ to the beginning of the circuit, it will have the same effect on the layers it moves through. Remembering that our initial state is the plus state, $\ket{+^n}$, this is an eigenstate of $X$ and so the residual operator disappears into the initial state:
\begin{align}
    \tilde{U}^X_{j,s} &= U_p \ldots U_{s+1} X_j U_s \ldots U_1\\
    & =U_p \ldots U_{s+1} \hat{U}_s \ldots \hat{U}_1  X_j \label{dragX}.
\end{align}
Therefore, injecting a single Pauli ($X$ or $Z$) is equivalent to running QAOA with a subset of the variational parameter signs flipped plus a residual Pauli that has no effect on the resulting distribution. 

$Y$ errors can be written as $-iZX$ and so we can propagate each through as though they act individually and drop the constant phase with the $Z$.

This is not to say that only qubits with errors are affected. This propagation results in flipped parameter signs on affected and possibly adjacent nodes that allow for the analysis in the following section. However, we can see that errors propagate as a light cone and affect all other qubits in that light cone (see Appendix \ref{appendix:lightcone}.)

\section{Results}

\subsection{Parameter dependent error resilience} \label{section:theorem}

Define the distance measure $\mathcal{D}(\tilde{U},U)$ as the total variation distance (TVD) between the distributions produced from the noisy and noiseless circuits:
\begin{align}\label{tvd}
    \mathcal{D}(\tilde{U},U) =  \| |\tilde U \ket{+^n}|^2 - | U \ket{+^n}|^2 \|_\text{tvd}.
\end{align}
This distance gives a measure of how much an error will alter the distribution induced by measuring the final state and an indication of how much this will affect the quality of the result. 

The following theorem bounds this distance in terms of the variational parameters.
\begin{theorem} \label{theorem}
Let $C(x)$ be a cost function on graph $\mathcal{G}$ with edge set $E$ and $n$ nodes, where node $i$ has degree $d_i$. 

We consider a QAOA circuit $U$ with $p$ layers and parameters $\gamma= \{\gamma_i\}, \beta = \{\beta_i\}$ for $i\in\{1,\ldots,p\}$. We denote the circuit with a single Pauli error $K \in \{X,Y,Z\}$ acting on qubit $j$ after layer $s$, by $\tilde{U}_{j,s}^K$. The total variation distance (TVD) in \eqref{tvd} is bounded by
\begin{align}
\mathcal{D}(\tilde{U}^K_{j,s},U)
  \leq 2\Big[\,
    &\chi_{K,X}\sum_{i=s+1}^{p}|\sin\beta_i| \\
  + &\chi_{K,Z}\max_{\omega\in\{0,\dots,d_j\}}\sum_{i=1}^{s}|\sin(\omega\gamma_i)|
  \,\Big] \notag\\
  \leq 2\Big[\,
    &\chi_{K,X}\sum_{i=s+1}^{p}|\beta_i|
  + \chi_{K,Z}\,d_j\sum_{i=1}^{s}|\gamma_i|
  \,\Big].\label{DK}
\end{align}
\end{theorem}
We defer the proof to Appendix \ref{appendix:theorem1}.

This is an upper bound on the distance between the two circuits; this is not a performance guarantee but an indication of the robustness of a set of QAOA parameters for a given graph. Specifically, a small value certifies robustness to noise but does not guarantee optimality of the resultant noisy state and a large bound does not prove worst-case error resilience but implies it is possible. 

From this, we see the distance between the circuits in the noiseless and noisy case depends on a subset of the variational parameters as well as the density of the affected graph node, with the particular subset depending on the type of noise present. A $Z$ error accumulates the angles associated with the mixer terms that come after it, meaning it has most effect when it happens early in the algorithm. Intuitively, the final measurement is in the computational basis and therefore immune to $Z$ errors. We would expect then that a $Z$ error has a less damaging effect when it occurs close to this measurement. Conversely, an $X$ error accumulates the cost term angles that precede it and thus has worse effect when it occurs later in the algorithm. Again intuitively, we start in the $\ket{+^n}$ state which is an $X$ eigenstate so we expect it to have less impact. A $Y$ error accumulates both, so is damaging at any point in the circuit. 

This can be seen in Fig.~\ref{fig:tvd} where the distance $\mathcal{D}(\tilde{U}^K_{0,s},U)$ is plotted for $K\in\{X,Y,Z\}$ and varying $s$. It shows that the dependence of the distance on $s$ is itself strongly dependent on $K$. A $Z$ error is most damaging when it occurs early in the circuit and an $X$ error is most damaging when it occurs near the end. The distance when a $Y$ error occurs closely follows the sum of the $X$ and $Z$ distances for all $s$, as expected from \eqref{DK}.

\begin{figure}[t]
    \centering
    \includegraphics[width=0.9\linewidth]{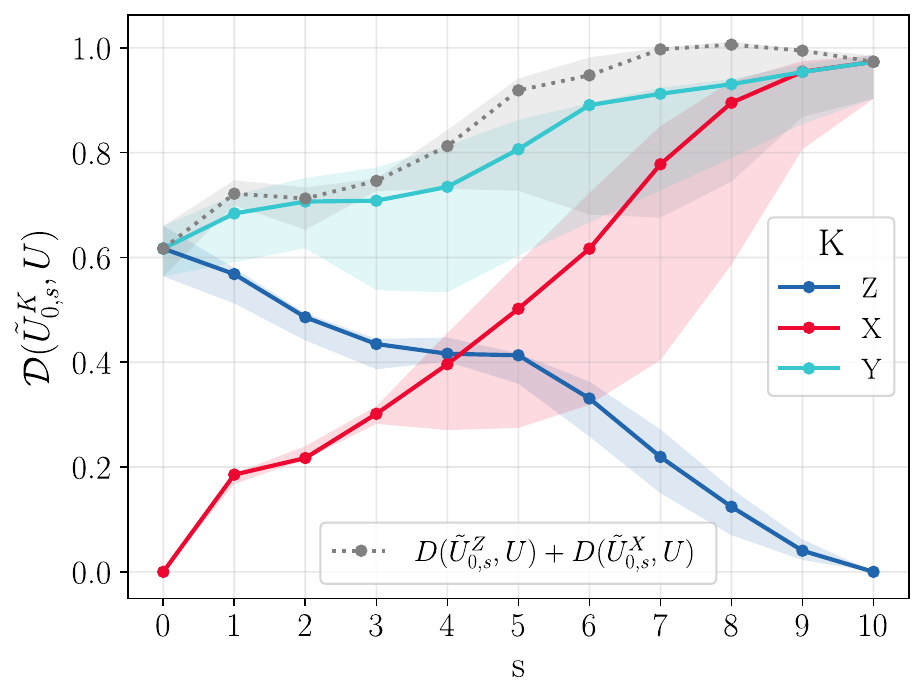}
    \caption{
     Total variation distance, $\mathcal{D}(\tilde{U}^K_{0,s},U)$, for QAOA with 10 layers and a single Pauli $K\in\{X,Y,Z\}$ injected after layer $s$ (where $s=0$ denotes the position before the first layer). We plot the median distance over 5 3-regular graphs with 10 nodes and the shaded region indicating the interquartile range. For each graph QAOA was run, using predicted parameters using the \textit{balanced method} from \cite{ieva}, without any errors then with a single Pauli injected and the distance between the distributions calculated using \eqref{tvd}. $Z$ errors are most damaging early in the circuit while $X$ errors are most damaging near the end. We also see that $\mathcal{D}(\tilde{U}^X_{0,s},U)+\mathcal{D}(\tilde{U}^Z_{0,s},U)$ closely tracks $\mathcal{D}(\tilde{U}^Y_{0,s},U)$.}
    \label{fig:tvd}
\end{figure}

The bound also implies that high degree qubits are more fragile to $X$ and $Y$ errors as the cost-term contribution picks up a $d_j$ dependence, whereas sensitivity to $Z$ errors is independent of degree. 

The size of the parameters then determines the size of the bound. Larger parameters push this upper bound above $1$ and hence we do not get a useful bound on the effect of errors. 
Smaller parameters keep the distance bounded below $1$ but reduce the expressibility of the circuit and the effectiveness of the algorithm. 

One may wonder whether this is equivalent to small angles preventing the circuit from moving far from the initial state and hence errors having little effect because of this. We provide numerical evidence in Appendix \ref{appendix:performance} that shows there is a separation between the regime where angles are too small to do anything and where the algorithm can remain error resilient. We also provide an upper bound on the distance the state can move from the initial state which gives further evidence for this separation. 

\begin{figure}[t]
    \centering
    \includegraphics[width=\linewidth]{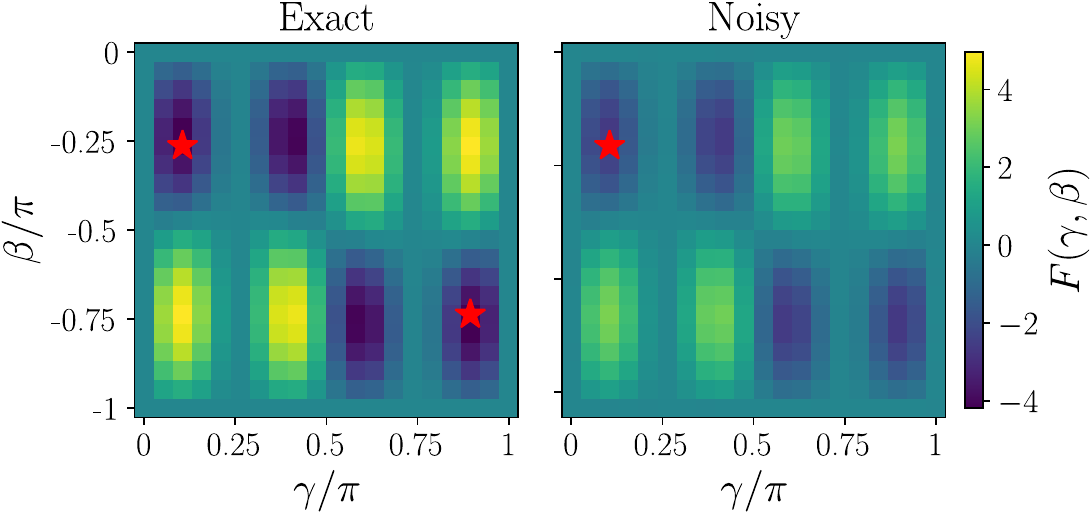}
    \caption{Noiseless and Noisy parameter landscape for 1 layer QAOA for Maxcut on 8 node 3-regular graph. Stars indicate global minima.}
    \label{fig:landscape}
\end{figure}

Fig.~\ref{fig:landscape} is an example of a $p=1$ QAOA parameter landscape for Max-Cut on an 8 node 3-regular graph. As expected, noise flattens the landscape, but also shifts the optimum toward smaller parameters, the behaviour Theorem \ref{theorem} predicts.

Practically, this gives insight into where different errors should be prioritized. Correcting $Z$ errors in early layers, $X$ in later layers and on high degree qubits helps to mitigate the effect of errors. Similarly, for variants of QAOA with independent angles for individual qubits e.g. MA-QAOA~\cite{herrman2022multi}, keeping parameters smaller on error prone or high degree qubits would increase robustness to errors. 

Although this theorem is stated for single Paulis, this can be extended to multiple errors simply by propagating all errors outwards such that the residual Paulis can be dropped. The distance is then bounded by the sum of all the phase flips caused by all propagations. 

Generalizing to the Pauli noise model as in \eqref{Pauli}, we define $\mathcal{E}(\tilde{U}^K_{j,s} ,U)$ as this upper bound on the distance
\begin{align}\label{upperbound}
    \mathcal{E}(\tilde{U}^K_{j,s} ,U)&=\chi_{K,X}\left[ 2\sum_{i = s+1}^{p}|\beta_i|\right] + \chi_{K,Z} \left[ 2d_j\sum_{i = 1}^{s}|\gamma_i|\right].
\end{align}
We find an upper bound on the expected distance by taking the expectation of \eqref{upperbound} for this error model, writing $\epsilon_{xy} = \epsilon_{x}+ \epsilon_{y}$ and $ \epsilon_{zy}=\epsilon_{y}+\epsilon_{z}$
\begin{align}
    \mathcal{E}(\gamma,\beta) &= \mathbb{E}_{\text{errors}}\left[\mathcal{E}(\tilde{U}^K_{j,s} ,U)\right]\\
    &=2n\sum_{i=1}^p
    \left[\epsilon_{zy}i|\beta_i| + 
    \epsilon_{xy}(p-i+1)d_{\text{avg}}|\gamma_i|\right].\label{mathcalE}
\end{align}
From this, we can infer which sets of parameters will provide more stability against general errors. The parameters in the final layers have more weight for $Z$ errors as it is always the parameters after the error that add to the distance. The converse is true for $X$ errors. This points toward ramping schedules, where $\gamma_i$ start small and increase and $\beta_i$ begin large in magnitude and decrease, which resemble adiabatic schedules and inherently suppress $\mathcal{E}(\gamma,\beta)$.

This gives us a quantity to measure the resilience of a circuit with parameters $(\gamma,\beta)$ to Pauli noise. The aim would then be to minimize this quantity while also minimizing the energy of the resulting state to balance the quality of the solution with robustness to noise. 

\subsubsection{Numerical results}

\begin{figure*}[t]
    \centering
    \includegraphics[width=\linewidth]{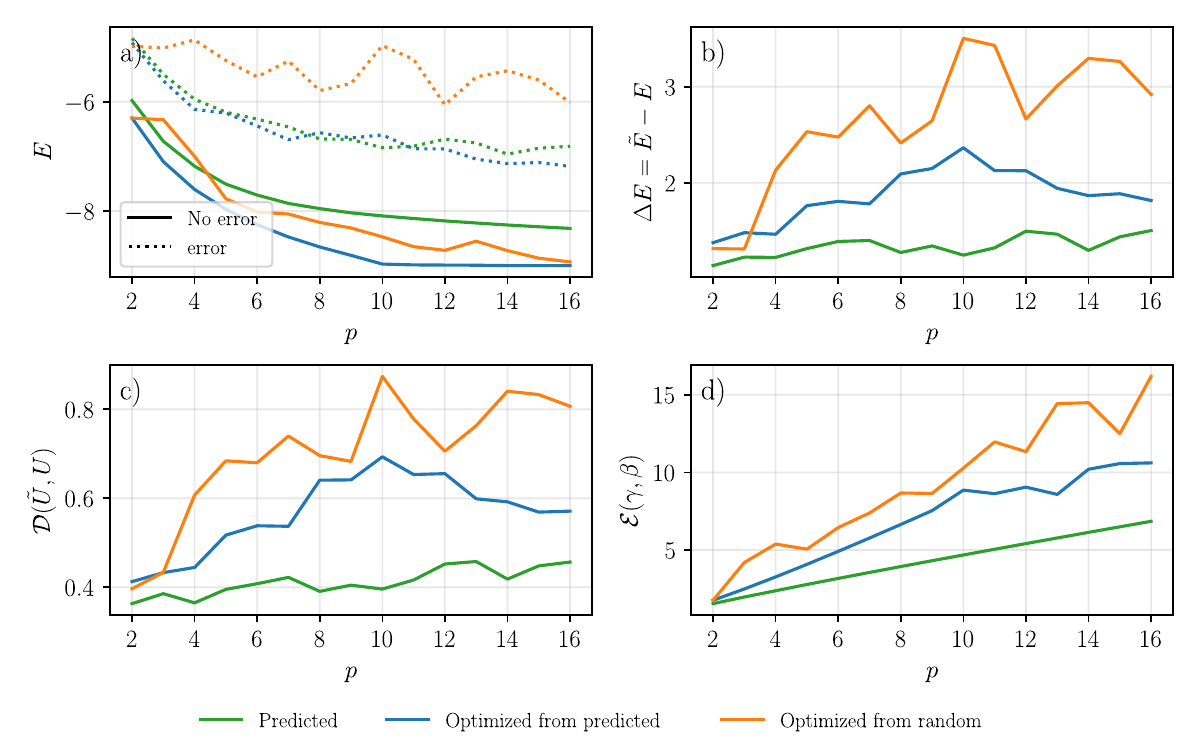}
    \caption{a) Noisy and noiseless energy, b) energy difference, c) distance and d) upper bound on expected distance for Max-Cut on a 10 node line graph using 3 different parameter generation methods: predicted using \textit{SKatan Method} from~\cite{ieva}; optimized starting from the predicted parameters; and optimized from a random starting point. Both optimizations were done offline, in exact simulation using LBFGS \cite{lbfgs}. Single Pauli $X$ and $Z$ were injected uniformly at a random layer and random qubit in the circuit and the results were averaged over 200 circuit runs.}
    \label{fig:dist}
\end{figure*}

To determine whether this upper bound gives useful insight into the performance of different choices of parameters, we compared 3 different choices of parameters for Max-Cut on a 10 node line graph for various numbers of QAOA layers $p$ and evaluated the noise resilience through the quantity $\mathcal{E}(\gamma,\beta)$. The predicted parameters are small, smooth, and monotonic whereas those optimized from random are large, jagged, and irregular. Noise was simulated using randomly injected Paulis after each QAOA layer.

The results in Fig.~\ref{fig:dist} show that not all parameters are equally error resilient. Though the parameters optimized from random perform reasonably well without any noise, they are the worst in the noisy case. Conversely, the predicted parameters perform worst in the noiseless case but perform similarly to optimized parameters in the noisy regime, perhaps suggesting optimization offers little benefit if noise is not taken into account. The difference in energy is tracked by the TVD \eqref{tvd}, justifying this as a good quantity to bound. We see a qualitative correlation between the upper bound on the distance and the actual distance. The predicted parameters maintain a much smaller distance from the noiseless state and the expected upper bound predicts this. Quantitatively, the bound is loose, but its qualitative agreement with $\mathcal{D}(\tilde U,U)$ and $\Delta E$ makes it a good qualitative indicator of relative noise resilience across the methods.

\subsection{Noise-informed parameter selection}\label{section:choosing}

Building on results in the previous section, we set out to devise methods for generating parameters to increase noise resilience and to determine which sets of existing parameter generation methods perform best in the noisy regime. In this section we will firstly present results showing that, depending on the noise rate, predetermined parameters can be more stable to noise than methods which require optimization. We will compare the stability of existing methods of parameter generation and introduce a noise-informed parameter optimization that we find to provide marginal improvements in results in the presence of noise. These preliminary results act as a proof of concept for future work exploring how to make algorithms robust in the NISQ era. Though this paper explores QAOA explicitly, it is possible that these results extend to more general VQAs.

\subsubsection{Stability of Predetermined Parameters}

We look at the approximation ratio achieved by different methods of parameter generation with varying rates of noise in the circuit. We define the Max-Cut approximation ratio 
\begin{align}
    AR(\gamma,\beta) = \frac{\langle C_{\text{cut}} \rangle_{\gamma,\beta}}{C^*_{\text{cut}}} =\frac{W-F(\gamma,\beta)}{W -C^*}.
\end{align}
where the optimal costs $C^*$ and $C^*_{\text{cut}}$ are defined above and $W = \sum_{(i,j)\in E}w_{ij}$.

From Fig.~\ref{fig:noise_rates} we see that at certain noise rates, optimizing parameters without noise is less effective than using these precomputed parameters or optimizing with noise. At low rates of noise, parameters optimized in exact emulation achieve the best approximation ratio across these graphs but beyond $\epsilon \approx 0.01$, optimizing either with knowledge of the noise or without does not improve the approximation ratio one could achieve by taking precomputed parameters without any optimization. This is significant as optimization can be costly, especially on device, therefore if it does not give advantage over precomputed parameters, skipping optimization altogether would yield notable cost savings.

\begin{figure}
    \centering
\includegraphics[width=\linewidth]{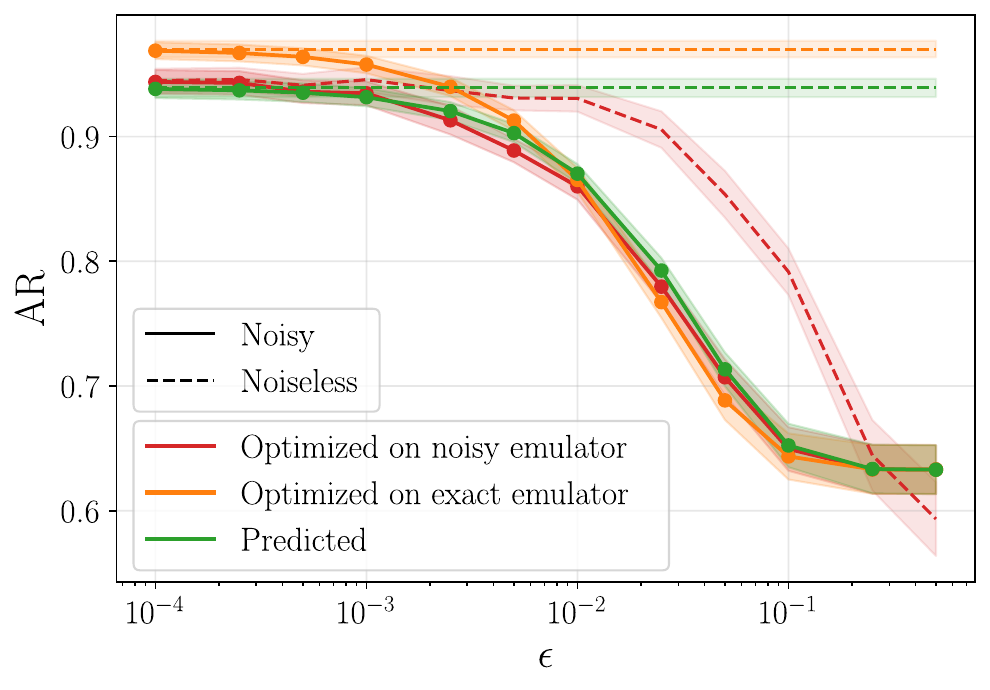}
    \caption{Mean approximation ratio achieved by QAOA with $p=4$ for Max-Cut on 6 Erd\H{o}s -- R\'enyi graphs with 6 nodes and Pauli noise channels inserted in the circuit with various noise rates. Three methods of parameter generation were employed: optimizing in exact emulation, optimizing in noisy emulation, and using unoptimized predicted parameters. Beyond a certain noise rate we see optimizing parameters offline is not beneficial. The shaded region indicates the standard error of the mean.}
    \label{fig:noise_rates}
\end{figure}

\subsubsection{Resilient Parameter optimization}\label{section:penalty}

Given the insight from Theorem ~\ref{theorem}, that smaller parameters are naturally more resilient to noise, we can attempt to encourage robustness by introducing penalties that give preference to smaller parameters.

We can add a penalty to the objective function based on the expected value of the upper bound on the distance between noisy and noiseless circuits in \eqref{mathcalE}. The parameter optimization problem then becomes
\begin{align}\label{penalty}
    (\gamma^*,\beta^*) = \argmin _{\gamma,\beta} \bigg[F(\gamma,\beta) + \lambda \;\mathcal{E}(\gamma,\beta) \bigg].
\end{align}
The penalty weight $\lambda$ should be chosen to balance the quality of the solution with the resilience of the parameters. 

Alternatively, since this penalty adds complexity to the optimization process that may be costly, a simpler way to restrict the optimizer to smaller parameters is to truncate the search space. For example, adding the constraints $0\leq\gamma_i\leq\frac{\pi}{2}$ and $-\frac{\pi}{2}\leq\beta_i\leq0$ for all $i$, bounds \eqref{mathcalE} by
\begin{align}
    \mathcal{E}(\gamma,\beta) \leq \frac{np\pi}{2}(p+1)(\epsilon_{zy}+d_{\text{avg}}\epsilon_{xy}).
\end{align}
Thus the optimization problem becomes
\begin{align}\label{truncation}
    (\gamma^*,\beta^*) = \argmin _{\gamma\in[0,\frac{\pi}{2}],\beta\in[-\frac{\pi}{2},0]}F(\gamma,\beta).
\end{align}
This constraint is simple to implement, but risks excluding a global minimum outside the bounded parameter space. 

The periodicity of QAOA parameters varies by problem type and hence the values for truncation may vary. Theorem 2 in \cite{lee2023depth} gives the periodicity of QAOA for unweighted Max-Cut as $2\pi$ in $\gamma$ and $\pi/2$ in $\beta$. Differing conventions in the mixer and cost Hamiltonians mean this translates to a period of $\pi$ in both $\gamma$ and $\beta$ here. Therefore the truncation in \eqref{truncation} shrinks each parameter space by half.

\begin{figure*}[t]
    \centering
    \includegraphics[width=0.9\linewidth]{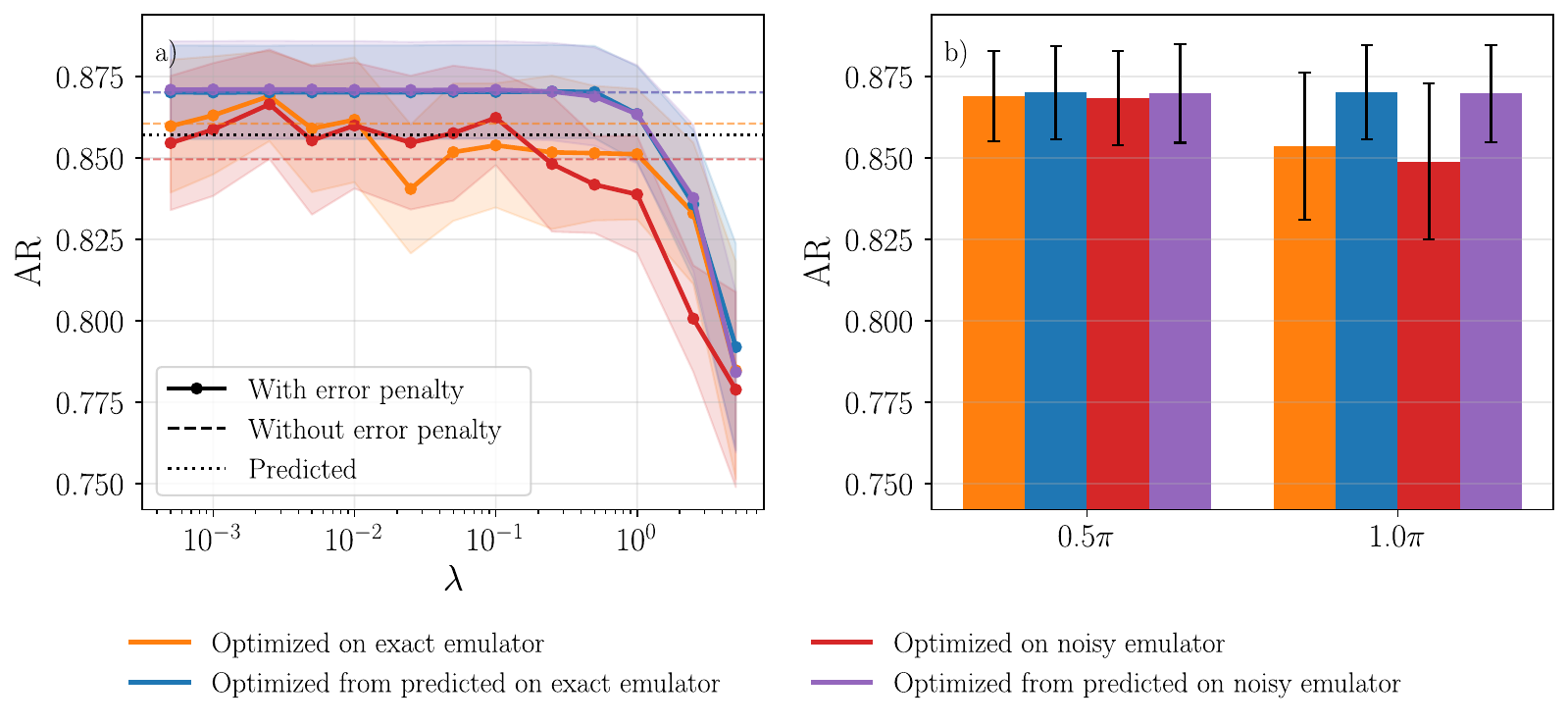}
    \caption{Noisy approximation ratios achieved by QAOA with parameters generated using noise-informed methods.
    QAOA  with $p=4$ was used for Max-Cut on 5 graphs with 6 nodes (2 Erd\H{o}s -- R\'enyi, 1 line, 2 3-regular) and these results are the average noisy approximation ratios over the graphs. Noise was injected using Pauli channels with error probability $\epsilon=0.01$ to calculate the resulting AR for each method. These plots show the mean over the 5 graphs and the shaded region show the standard error of the mean. Parameter optimization was done either using a noisy or exact emulator and starting from both random and precomputed parameters. The precomputed parameters were calculated using the \textit{balanced method} from \cite{ieva}.
    a) shows results from parameters optimized using the cost function in \eqref{penalty} with varying $\lambda$ values. b) shows results from parameters optimized using the cost landscape truncation in \eqref{truncation}, where the optimizer was restricted to search parameters with absolute value up to $\pi$ and $\pi/2$. Results for each graph are given in Appendix ~\ref{appendix:penalty}.}
    \label{fig:penalty-plot-avg}
\end{figure*}

Both of these techniques were tested on 5 graphs with 6 nodes using 4 different optimization methods. Two methods in exact emulation with LBFGS \cite{lbfgs} optimizer and two optimized with SPSA \cite{spsa} with noise injected in the circuit via Pauli channels with probability $\epsilon=0.01$. 

The average results are shown in Fig.~\ref{fig:penalty-plot-avg}, full results for each graph can be found in Appendix ~\ref{appendix:penalty}. The results on the effectiveness of these techniques are inconclusive. While there is marginal improvement for some values of $\lambda$ for methods not optimizing from predicted parameters, results are not consistent enough to draw conclusions on whether this will help improve results under noise in general. This has also only been tested on very small graphs due to costly noisy emulation so there is no indication of how this would scale.

\subsection{Adiabatic limit of QAOA}\label{section:adiabatic}

The sets of parameters that best minimize \eqref{mathcalE} while still maintaining good performance, are those that resemble adiabatic schedules. That is, the cost parameter ramps up and the mixer parameter ramps down. This inspired an extension of this analysis to the QAA to gain insight into the error resilience of the QAA, or into potential QAOA parameter schedules that perform well in the presence of errors.

We approximate the QAA using QAOA by taking the limit $p\rightarrow\infty$. Here, the layer index becomes a continuous time, $i\rightarrow t \in [0,T]$ and the discrete parameters become functions of time $\gamma_i \rightarrow \gamma(t)$, $\beta_i \rightarrow \beta(t)$. In the standard QAA, as in \eqref{sym}, there is a single schedule $s(t)$ such that $\gamma(t)=1-\beta(t)=s(t)$. Using this, \eqref{mathcalE} becomes
\begin{align}
    \mathcal{E}(s)
    &\approx 2n\int_0^T\epsilon_{zy}t|1-s(t)| + \epsilon_{xy}(T-t)d_{\text{avg}}|s(t)|\,\mathrm{d}t\label{expected}.
\end{align}
We want to find whether this is bounded, in order to learn about the intrinsic error resilience of these schedules.

More generally we want to know if we can do better than this with any schedule. The simplest way to do this is to assume unphysical schedules to upper and lower bound this such that any physical schedule sits between them.

Assuming $0\leq s(t) \leq 1$ and $ds/dt\geq0$, we bound this quantity by choosing extremal values for $s(t)$.

$\mathcal{E}(s)$ is maximized by the constant schedules,
$s(t)=0$ or $s(t)=1$, whichever is larger depending on the error rates. If instead $s(t)$ is a step function with a step at $t=\frac{\epsilon_{xy}d_\text{avg}T}{\epsilon_{xy}d_\text{avg}+\epsilon_{zy}}$, evaluating \eqref{expected} gives a lower bound. Therefore, we have 
\begin{align}
nT^2 \cdot
\frac{\epsilon_{xy}\epsilon_{zy}d_{\text{avg}}}
     {\epsilon_{zy} + \epsilon_{xy}d_{\text{avg}}}
\leq
\mathcal{E}(s)
\leq
nT^2 \cdot\max\{\epsilon_{xy}d_{\text{avg}},\, \epsilon_{zy}\}
\end{align}
These are unphysical schedules that lower bound $\mathcal{E}(s)$ such that a physical schedule can only do worse. This gives it a fixed floor that scales with $nT^2$, to suppress this below $\mathcal{O}(1)$ requires error rates $\epsilon\lesssim 1/(nT^2)$ which is small once $T$ is fixed by the adiabatic runtime $\sim\sqrt{N}$. This indicates a fixed resilience to errors that cannot be improved if we maintain a single adiabatic schedule. 

To attempt to beat the lower bound on $\mathcal{E}(s)$, we introduce two independent schedules. Instead of $s(t)$ in \eqref{sym}, we have $\gamma(t),\beta(t)$
\begin{align}
    H(t) = \gamma(t) H_{\text{final}} + \beta (t) H_{\text{initial}}.
\end{align}
The symmetric constraint $\gamma(t) = 1-\beta(t)$ limited how small \eqref{expected} could get, since demanding $\gamma = \beta$ at the crossing forces the schedules to meet at a crossing point $\delta = 1/2$. Using individual schedules $\gamma(t)$ and $\beta(t)$, removes this restriction and $\delta$ can be pushed below $1/2$ such that the integral is reduced. Therefore, instead of \eqref{expected} we get
\begin{align}
    \mathcal{E}_{\text{cont}}(\gamma,\beta)
    &\approx 2n\int_0^T\epsilon_{zy}t|\beta(t)| + \epsilon_{xy}(T-t)d_{\text{avg}}|\gamma(t)|\,\mathrm{d}t\label{expected2}.
\end{align}
Again, by assuming unphysical schedules that do not obey the adiabatic constraints, this bound can go to zero, namely individual step functions where $\beta(0)=1$ but $\beta(t>0)=0$ and $\gamma(0\leq t<T)=0$ but $\gamma(T)=1$. 

This frees us from the lower bound placed on symmetric schedules. However, since we assumed unphysical schedules it does not guarantee a physically possible schedule with a better achievable $\mathcal{E}_{\text{cont}}(\gamma,\beta)$ exists. 

This raises the question of whether the two schedule approach has better inherent error performance than the single schedule when physical constraints are imposed, in particular, when we decrease the crossing point $\delta$.

To investigate this we use the adiabatic theorem to examine physically realistic schedules via the example problem of unstructured quantum search. This work follows a similar structure to~\cite{roland2002quantum} and~\cite{isermann2021optimal}, but with two schedules. 

We find that when we build in physically realistic constraints on the schedules and runtime, it is not possible to get significantly better error behaviour. Physical constraints limit the speed at which the schedules can change and hence determine the total runtime of the algorithm. Forcing the schedules smaller increases the runtime. Increasing the duration of the algorithm has more impact on $\mathcal{E}$ than shrinking the schedules and hence we find worse error behaviour.

\subsubsection{Unstructured quantum Search}

Quantum search for element $m$ using the adiabatic evolution in \eqref{schro} is given by the Hamiltonian
\begin{align}\label{ham}
    H(t) = \gamma(t) H_m + \beta (t) H_0,
\end{align}
where
\begin{align}
    H_m &= I - \ket{m}\bra{m}\\
    H_0 &= I - \ket{+^n}\bra{+^n}.\label{H0}
\end{align}
We impose on the two schedules, assumptions generally made for $s(t)$. Firstly, boundary conditions $\gamma(0) = 0$, $\beta(0) = 1$, $\gamma(T) = 1$,$\beta(T) = 0$. Secondly, the monotonicity of each schedule, meaning they cross at a single point where $\gamma = \beta = \delta$ and that they cannot both vanish simultaneously. Finally, that $\gamma(t)$ and $\beta(t)$ are independent, in particular $\gamma \neq 1 - \beta$. Looking at the form of \eqref{expected2}, and taking inspiration from typical adiabatic schedules, we might assume that a standard adiabatic schedule, with the crossing point $\delta$ pushed lower, would increase the resilience to errors. 

Under these assumptions, for quantum search with 2 schedules the gap $g(t)$ is 
\begin{align}
    g(t)  = \sqrt{(\beta(t)-\gamma(t))^2+\frac{4\beta(t)\gamma(t)}{N}},
\end{align}
where $N = 2^n$ and the numerator on the left hand side of \eqref{localadiabatic} becomes 
\begin{align}
    \big|\bra{E_1;t}\frac{dH}{dt}\ket{E_0;t}\big|=\frac{\sqrt{N-1}}{Ng(t)}|\frac{d\gamma}{dt}\beta - \frac{d\beta}{dt}\gamma|,
\end{align}
the proofs of which are in Appendix ~\ref{appendix:Adiabatic}.

So, using the monotonicity condition, \eqref{localadiabatic} becomes
\begin{equation}
    \frac{d\gamma}{dt}\beta - \frac{d\beta}{dt}\gamma \leq \eta\sqrt{N}\left[(\gamma - \beta)^2 + \frac{4\gamma\beta}{N}\right]^{3/2},
    \label{constraint}
\end{equation}
giving the condition for local adiabatic evolution for \eqref{ham}. Using this, we determine a bound on the total time of the algorithm and therefore bound $\mathcal{E}_{\text{cont}}(\gamma,\beta)$ to evaluate the noise resilience of the two schedule case for quantum search. 

We split the analysis at $(\gamma-\beta)^2= 4\gamma\beta/N$, treating the regime where
$\gamma$ and $\beta$ are far apart separately from the regime near the crossing point $\delta$. 

In the crossing region $(\gamma-\beta)^2 < 4\gamma\beta/N$, and the $\frac{4\gamma\beta}{N}$ term dominates \eqref{constraint} and we get
\begin{align}
     \frac{d\gamma}{dt}\beta-\frac{d\beta}{dt}\gamma \leq \sqrt{N}\eta\bigg(\frac{4\gamma\beta}{N}\bigg)^{\frac{3}{2}}= \frac{8\eta(\gamma\beta)^{\frac{3}{2}}}{N}.
\end{align}

The region is defined by $|\gamma - \beta| < 2\sqrt{\gamma\beta}/\sqrt{N}$, which near the crossing point spans $-2\delta/\sqrt{N} \lesssim \gamma - \beta \lesssim 2\delta/\sqrt{N}$, giving a width
$\Delta \lesssim 4\delta/\sqrt{N}$. At the crossing $\gamma=\beta=\delta$, the rate at which the schedule traverses the region is bounded
\begin{equation}
    \frac{d\gamma}{dt}-\frac{d\beta}{dt}\leq\frac{8\eta\delta^2}{N}.
    \label{rate}
\end{equation}
which in fact holds throughout the crossing region (proof in Appendix ~\ref{appendix:rate_crossing}). The time spent in this region is therefore 
\begin{equation}
    T_{\text{cross}}\gtrsim \frac{\Delta}{\max(\frac{d\gamma}{dt} - \frac{d\beta}{dt})}\approx \frac{4\delta/\sqrt{N}}{8\eta\delta^2/N}
    = \frac{\sqrt{N}}{2\eta\delta}.
\end{equation}

We can lower bound the total time $T$ by the time spent in the crossing region $T_{\text{cross}}$. The total runtime therefore is $T \geq T_{\text{cross}}= \Omega\left(\frac{\sqrt{N}}{\eta\delta}\right)$.

We now lower bound \eqref{expected2} by evaluating it within this region. Let $[A,B]$ be the time interval during which the schedule is within the crossing region, so that $B-A=T_{\text{cross}}$. The crossing region width means $\gamma, \beta \geq \delta(1 - 2/\sqrt{N})$. Restricting \eqref{expected2} to $[A,B]$ and using this lower bound gives
\begin{align}
    \mathcal{E}_{\text{cont}}(\gamma,\beta)
    &\gtrsim 2n\delta\left(1 - \frac{2}{\sqrt{N}}\right) \int_A^B \epsilon_{zy}t + \epsilon_{xy}d_{\text{avg}}(T-t)\mathrm{d}t .
\end{align}
Writing $B^2 - A^2 = (A+B)T_{\text{cross}}$, we bound both terms in the integral separately
\begin{align}
    \int_A^B \epsilon_{zy}\,t\,\mathrm{d}t
    &= \tfrac{1}{2}\epsilon_{zy}(A+B)\,T_{\text{cross}}
    \geq \tfrac{1}{2}\epsilon_{zy}\,T_{\text{cross}}^2 ,\\
    \int_A^B \epsilon_{xy}d_{\text{avg}}(T-t)\,\mathrm{d}t
    &= \epsilon_{xy}d_{\text{avg}}\big[T-\tfrac{1}{2}(B+A)\big]\,T_{\text{cross}},\\&
    \geq \tfrac{1}{2}\epsilon_{xy}d_{\text{avg}}\,T_{\text{cross}}^2 ,
\end{align}
where the first inequality uses $A+B \geq B-A = T_{\text{cross}}$, and the second uses $(T-A)+(T-B) \geq T-A \geq B-A = T_{\text{cross}}$ since $B \leq T$. Together these give
\begin{align}
    \mathcal{E}_{\text{cont}}
    &\gtrsim n\delta\left(1 - \frac{2}{\sqrt{N}}\right)\left(\epsilon_{zy} + \epsilon_{xy}d_{\text{avg}}\right) T_{\text{cross}}^2 .
\end{align}
Substituting $T_{\text{cross}} \gtrsim \frac{\sqrt{N}}{2\eta\delta}$ then gives
\begin{align}
    \mathcal{E}_{\text{cont}}
    &= \Omega\!\left( \frac{nN}{\eta^2\delta}\left(\epsilon_{zy} + \epsilon_{xy}d_{\text{avg}}\right)\right).
\end{align}
This bound grows as $\delta$ decreases. A smaller crossing point lengthens $T_{\text{cross}}$ and the $T_{\text{cross}}^2$ dependence of $\mathcal{E}(\gamma,\beta)$ means the net effect is worse the smaller $\delta$ is.

This intuitively makes sense; a longer runtime $T$ with fixed noise rates means more errors. This is analogous to increasing the number of layers $p$ in a QAOA circuit so we would expect more errors.

\section{Conclusion}

We find that the choice of variational parameters in the QAOA affects the stability of the circuit in the presence of noise. In fact, using a Pauli error model with single injected Pauli errors, we derive a bound on how far the noisy output state is from the desired noiseless state based on the sum of the parameters. This allows us in turn to compare the stability of methods for generating parameters as well as adding penalties to the optimization method in order to construct an error resilient circuit.

We find that optimizing without consideration of noise (which can be done to avoid costly noisy optimization on device \cite{streif2020training,pelofske2024scaling,maciejewski2026quantum}) can perform worse in the noisy case than unoptimized parameters.

In particular, beyond a certain noise rate precomputed parameters can perform better than optimization methods. Verifying that this scales and is applicable to real hardware noise is left to future work. However, if this shows that optimizing is not beneficial, using precomputed parameters would save optimization costs.

Further work is needed to verify the methods presented here for noise-informed parameter generation that attempt to manufacture error resilience are general and effective and scale beyond small instances. The preliminary results show simply optimizing starting from predicted parameters is the most effective method in this case. 

The stability of non-optimized parameters comes primarily from their shape as the adiabatic-looking schedules -- monotonically increasing $|\gamma|$ and decreasing $|\beta|$ -- tend to minimize the bound and exhibit more resilience to errors. This led us to investigate the adiabatic algorithm in the context of our bound to see if we gain any insight about which are the best QAOA schedules. We found that the single schedule adiabatic algorithm has a fixed floor on how error resilient it can be based on our bound. In an attempt to remove this lower bound we derive a 2 schedule adiabatic algorithm and show this does not help in improving the error resilience of the algorithm.

The bound \eqref{mathcalE} shows us that when the most connected qubits (high $d_j$) have errors, the impact is worse. This may be combated by an MA-QAOA style parameter shrinking on these qubits to counteract errors on them, essentially pre-preparing them for noise. Or similarly, a form of device-dependent parameter scaling may be possible. For example, if certain qubits on device are known to have high error rates, having individual parameters for these and scaling them down to reduce the impact would likely improve the error resilience of the circuit. The question of how this affects performance is open.

Furthermore, the structure of the derivation of this bound means it may be possible to extend this to other VQAs. This bound works because we have two anticommuting alternating operators and errors in the same anticommuting basis, therefore any similarly structured ansatz may have similar behaviour.

Finally, having found that parameters generated without optimization often perform better than parameters optimized without noise, it is an open question as to whether a method of noise-informed parameter prediction using insights from this bound would improve this even further. 

\section{Acknowledgments}

We would like to thank Ieva \v{C}epait\.{e} and Ra\'ul Garc\'ia-Patr\'on for helpful discussions on the topic of this paper, as well as Niam Vaishnav for assistance with numerical implementations. We also thank Sheila P\'erez Garc\'ia, Josh Blake, Ioannis Kolotouros, Elizabeth Crosson and other members of the Phasecraft team for helpful feedback on these results. 

We acknowledge the use of large language models (Claude Opus 4.8 and 5) to assist with numerics, to check proofs devised by the authors, and to refine the writing of this paper. We verified all AI-assisted outputs and take full responsibility for all content.

\bibliographystyle{apsrev4-2}
\bibliography{refs}

@book{korte2011combinatorial,
  title={Combinatorial optimization: Theory and Algorithms},
  author={Korte, Bernhard H and Vygen, Jens},
  volume={1},
  year={2011},
  publisher={Springer}
}

@article{farhi,
  title={A quantum approximate optimization algorithm},
  author={Farhi, Edward and Goldstone, Jeffrey and Gutmann, Sam},
  journal={arXiv preprint arXiv:1411.4028},
  year={2014}
}

@book{nielsen2010quantum,
  title     = {Quantum Computation and Quantum Information:
               10th Anniversary Edition},
  author    = {Nielsen, Michael A. and Chuang, Isaac L.},
  year      = {2010},
  publisher = {Cambridge University Press},
  address   = {Cambridge}
}

@article{qaa,
  title={Quantum computation by adiabatic evolution},
  author={Farhi, Edward and Goldstone, Jeffrey and Gutmann, Sam and Sipser, Michael},
  journal={arXiv preprint quant-ph/0001106},
  year={2000}
}

@article{stilck2021limitations,
  title={Limitations of optimization algorithms on noisy quantum devices},
  author={Stilck Fran{\c{c}}a, Daniel and Garcia-Patron, Raul},
  journal={Nature Physics},
  volume={17},
  number={11},
  pages={1221--1227},
  year={2021},
  publisher={Nature Publishing Group UK London}
}

@article{gonzalez2022error,
  title={Error propagation in NISQ devices for solving classical optimization problems},
  author={Gonz{\'a}lez-Garc{\'\i}a, Guillermo and Trivedi, Rahul and Cirac, J Ignacio},
  journal={PRX Quantum},
  volume={3},
  number={4},
  pages={040326},
  year={2022},
  publisher={APS}
}

@article{roland2002quantum,
  title={Quantum search by local adiabatic evolution},
  author={Roland, J{\'e}r{\'e}mie and Cerf, Nicolas J},
  journal={Physical Review A},
  volume={65},
  number={4},
  pages={042308},
  year={2002},
  publisher={APS}
}

@article{isermann2021optimal,
  title={On the optimal schedule of adiabatic quantum computing},
  author={Isermann, Stefan},
  journal={Quantum Information Processing},
  volume={20},
  number={9},
  pages={300},
  year={2021},
  publisher={Springer}
}

@article{twirling3,
  title = {Randomized Compiling for Scalable Quantum Computing on a Noisy Superconducting Quantum Processor},
  author = {Hashim, Akel and Naik, Ravi K. and Morvan, Alexis and Ville, Jean-Loup and Mitchell, Bradley and Kreikebaum, John Mark and Davis, Marc and Smith, Ethan and Iancu, Costin and O'Brien, Kevin P. and Hincks, Ian and Wallman, Joel J. and Emerson, Joseph and Siddiqi, Irfan},
  journal = {Phys. Rev. X},
  volume = {11},
  issue = {4},
  pages = {041039},
  numpages = {12},
  year = {2021},
  month = {Nov},
  publisher = {American Physical Society},
  doi = {10.1103/PhysRevX.11.041039},
  url = {https://link.aps.org/doi/10.1103/PhysRevX.11.041039}
}

@article{twirling,
  title = {Noise tailoring for scalable quantum computation via randomized compiling},
  author = {Wallman, Joel J. and Emerson, Joseph},
  journal = {Phys. Rev. A},
  volume = {94},
  issue = {5},
  pages = {052325},
  numpages = {9},
  year = {2016},
  month = {Nov},
  publisher = {American Physical Society},
  doi = {10.1103/PhysRevA.94.052325},
  url = {https://link.aps.org/doi/10.1103/PhysRevA.94.052325}
}

@article{twirling2,
  title={Constructing smaller Pauli twirling sets for arbitrary error channels},
  author={Cai, Zhenyu and Benjamin, Simon C},
  journal={Scientific reports},
  volume={9},
  number={1},
  pages={11281},
  year={2019},
  publisher={Nature Publishing Group UK London}
}

@article{VQA,
  title={Variational quantum algorithms},
  author={Cerezo, Marco and Arrasmith, Andrew and Babbush, Ryan and Benjamin, Simon C and Endo, Suguru and Fujii, Keisuke and McClean, Jarrod R and Mitarai, Kosuke and Yuan, Xiao and Cincio, Lukasz and others},
  journal={Nature Reviews Physics},
  volume={3},
  number={9},
  pages={625--644},
  year={2021},
  publisher={Nature Publishing Group UK London}
}

@article{BarrenPlateaus,
  title={Noise-induced barren plateaus in variational quantum algorithms},
  author={Wang, Samson and Fontana, Enrico and Cerezo, Marco and Sharma, Kunal and Sone, Akira and Cincio, Lukasz and Coles, Patrick J},
  journal={Nature communications},
  volume={12},
  number={1},
  pages={6961},
  year={2021},
  publisher={Nature Publishing Group UK London}
}

@article{BP,
  title={Barren plateaus in variational quantum computing},
  author={Larocca, Martin and Thanasilp, Supanut and Wang, Samson and Sharma, Kunal and Biamonte, Jacob and Coles, Patrick J and Cincio, Lukasz and McClean, Jarrod R and Holmes, Zo{\"e} and Cerezo, Marco},
  journal={Nature Reviews Physics},
  pages={1--16},
  year={2025},
  publisher={Nature Publishing Group UK London}
}

@book{lidarandbrun,
  title={Quantum error correction},
  author={Lidar, Daniel A and Brun, Todd A},
  year={2013},
  publisher={Cambridge university press}
}

@article{leo,
  title={Quantum approximate optimization algorithm: Performance, mechanism, and implementation on near-term devices},
  author={Zhou, Leo and Wang, Sheng-Tao and Choi, Soonwon and Pichler, Hannes and Lukin, Mikhail D},
  journal={Physical Review X},
  volume={10},
  number={2},
  pages={021067},
  year={2020},
  publisher={APS}
}

@article{warmstarts,
  title = {Variational Quantum Simulation: A Case Study for Understanding Warm Starts},
  author = {Puig, Ricard and Drudis, Marc and Thanasilp, Supanut and Holmes, Zo\"e},
  journal = {PRX Quantum},
  volume = {6},
  issue = {1},
  pages = {010317},
  numpages = {39},
  year = {2025},
  month = {Jan},
  publisher = {American Physical Society},
  doi = {10.1103/PRXQuantum.6.010317},
  url = {https://link.aps.org/doi/10.1103/PRXQuantum.6.010317}
}

@article{nphard,
  title={On the computational complexity of combinatorial problems},
  author={Karp, Richard M.},
  journal={Networks},
  volume={5},
  number={1},
  pages={45--68},
  year={1975},
  publisher={Wiley Online Library}
}

@article{ieva,
  title={Quantum-Enhanced Optimization by Warm Starts},
  author={{\v{C}}epait{\.e}, Ieva and Vaishnav, Niam and Zhou, Leo and Montanaro, Ashley},
  journal={arXiv preprint arXiv:2508.16309},
  year={2025}
}

@article{NISQ,
  doi = {10.22331/q-2018-08-06-79},
  url = {https://doi.org/10.22331/q-2018-08-06-79},
  title = {Quantum {C}omputing in the {NISQ} era and beyond},
  author = {Preskill, John},
  journal = {{Quantum}},
  issn = {2521-327X},
  publisher = {{Verein zur F{\"{o}}rderung des Open Access Publizierens in den Quantenwissenschaften}},
  volume = {2},
  pages = {79},
  month = aug,
  year = {2018}
}

@article{vqe,
  title={A variational eigenvalue solver on a photonic quantum processor},
  author={Peruzzo, Alberto and McClean, Jarrod and Shadbolt, Peter and Yung, Man-Hong and Zhou, Xiao-Qi and Love, Peter J and Aspuru-Guzik, Al{\'a}n and O’brien, Jeremy L},
  journal={Nature communications},
  volume={5},
  number={1},
  pages={4213},
  year={2014},
  publisher={Nature Publishing Group UK London}
}

@article{vqas,
  title={The theory of variational hybrid quantum-classical algorithms},
  author={McClean, Jarrod R and Romero, Jonathan and Babbush, Ryan and Aspuru-Guzik, Al{\'a}n},
  journal={New Journal of Physics},
  volume={18},
  number={2},
  pages={023023},
  year={2016},
  publisher={IOP Publishing}
}

@article{qaoagood,
  title={Variational quantum algorithms for combinatorial optimization},
  author={Perez-Ramirez, Daniel F},
  journal={arXiv preprint arXiv:2407.06421},
  year={2024}
}

@article{qaoaevidence,
  title={Solving boolean satisfiability problems with the quantum approximate optimization algorithm},
  author={Boulebnane, Sami and Montanaro, Ashley},
  journal={PRX Quantum},
  volume={5},
  number={3},
  pages={030348},
  year={2024},
  publisher={APS}
}

@article{qaoaevidence2,
  title={Quantum supremacy through the quantum approximate optimization algorithm},
  author={Farhi, Edward and Harrow, Aram W},
  journal={arXiv preprint arXiv:1602.07674},
  year={2016}
}

@article{qaoaevidence3,
  doi = {10.22331/q-2022-07-07-759},
  url = {https://doi.org/10.22331/q-2022-07-07-759},
  title = {The {Q}uantum {A}pproximate {O}ptimization {A}lgorithm and the {S}herrington-{K}irkpatrick {M}odel at {I}nfinite {S}ize},
  author = {Farhi, Edward and Goldstone, Jeffrey and Gutmann, Sam and Zhou, Leo},
  journal = {{Quantum}},
  issn = {2521-327X},
  publisher = {{Verein zur F{\"{o}}rderung des Open Access Publizierens in den Quantenwissenschaften}},
  volume = {6},
  pages = {759},
  month = jul,
  year = {2022}
}

@article{qaoaforsim,
  title={Approximating Ground States of Quantum Hamiltonians with Snapshot-QAOA},
  author={Tate, Reuben and Langfitt, Quinn and Pelofske, Elijah and Kirmani, Ammar and B{\"a}rtschi, Andreas and Golden, John and Eidenbenz, Stephan},
  journal={arXiv preprint arXiv:2412.17990},
  year={2024}
}

@article{optop,
  title={Challenges and opportunities in quantum optimization},
  author={Abbas, Amira and Ambainis, Andris and Augustino, Brandon and B{\"a}rtschi, Andreas and Buhrman, Harry and Coffrin, Carleton and Cortiana, Giorgio and Dunjko, Vedran and Egger, Daniel J and Elmegreen, Bruce G and others},
  journal={Nature Reviews Physics},
  pages={1--18},
  year={2024},
  publisher={Nature Publishing Group}
}

@article{GW,
  title={Improved approximation algorithms for maximum cut and satisfiability problems using semidefinite programming},
  author={Goemans, Michel X and Williamson, David P},
  journal={Journal of the ACM (JACM)},
  volume={42},
  number={6},
  pages={1115--1145},
  year={1995},
  publisher={ACM New York, NY, USA}
}

@article{GW1,
author = {Karloff, Howard},
title = {How Good is the Goemans--Williamson MAX CUT Algorithm?},
journal = {SIAM Journal on Computing},
volume = {29},
number = {1},
pages = {336-350},
year = {1999},
doi = {10.1137/S0097539797321481}
}

@article{marshall2020characterizing,
  title={Characterizing local noise in QAOA circuits},
  author={Marshall, Jeffrey and Wudarski, Filip and Hadfield, Stuart and Hogg, Tad},
  journal={Iop Scinotes},
  volume={1},
  number={2},
  pages={025208},
  year={2020},
  publisher={IOP Publishing}
}

@article{xue2021effects,
  title={Effects of quantum noise on quantum approximate optimization algorithm},
  author={Xue, Cheng and Chen, Zhao-Yun and Wu, Yu-Chun and Guo, Guo-Ping},
  journal={Chinese Physics Letters},
  volume={38},
  number={3},
  pages={030302},
  year={2021},
  publisher={Chinese Physical Society and IOP Publishing Ltd}
}

@article{gradient,
  title={The effect of classical optimizers and Ansatz depth on QAOA performance in noisy devices},
  author={Pellow-Jarman, Aidan and McFarthing, Shane and Sinayskiy, Ilya and Park, Daniel K and Pillay, Anban and Petruccione, Francesco},
  journal={Scientific reports},
  volume={14},
  number={1},
  pages={16011},
  year={2024},
  publisher={Nature Publishing Group UK London}
}

@inproceedings{reinfo,
  title={Learning to optimize variational quantum circuits to solve combinatorial problems},
  author={Khairy, Sami and Shaydulin, Ruslan and Cincio, Lukasz and Alexeev, Yuri and Balaprakash, Prasanna},
  booktitle={Proceedings of the AAAI conference on artificial intelligence},
  volume={34},
  number={3},
  pages={2367--2375},
  year={2020}
}

@article{vqes,
  title={The variational quantum eigensolver: a review of methods and best practices},
  author={Tilly, Jules and Chen, Hongxiang and Cao, Shuxiang and Picozzi, Dario and Setia, Kanav and Li, Ying and Grant, Edward and Wossnig, Leonard and Rungger, Ivan and Booth, George H and others},
  journal={Physics Reports},
  volume={986},
  pages={1--128},
  year={2022},
  publisher={Elsevier}
}

@article{ugc,
  title={Optimal inapproximability results for {MAX-CUT} and other 2-variable CSPs?},
  author={Khot, Subhash and Kindler, Guy and Mossel, Elchanan and O’Donnell, Ryan},
  journal={SIAM Journal on Computing},
  volume={37},
  number={1},
  pages={319--357},
  year={2007},
  publisher={SIAM}
}

@article{hardparams,
  title={Training variational quantum algorithms is NP-hard},
  author={Bittel, Lennart and Kliesch, Martin},
  journal={Physical review letters},
  volume={127},
  number={12},
  pages={120502},
  year={2021},
  publisher={APS}
}

@article{resilience_VQEs,
  title={Noise Resilience and Robust Convergence Guarantees for the Variational Quantum Eigensolver},
  author={Legnini, Mirko and Berberich, Julian},
  journal={arXiv preprint arXiv:2601.16758},
  year={2026}
}

@article{variants,
  title={A review on quantum approximate optimization algorithm and its variants},
  author={Blekos, Kostas and Brand, Dean and Ceschini, Andrea and Chou, Chiao-Hui and Li, Rui-Hao and Pandya, Komal and Summer, Alessandro},
  journal={Physics Reports},
  volume={1068},
  pages={1--66},
  year={2024},
  publisher={Elsevier}
}

@article{fontana2021evaluating,
  title={Evaluating the noise resilience of variational quantum algorithms},
  author={Fontana, Enrico and Fitzpatrick, Nathan and Ramo, David Mu{\~n}oz and Duncan, Ross and Rungger, Ivan},
  journal={Physical Review A},
  volume={104},
  number={2},
  pages={022403},
  year={2021},
  publisher={APS}
}

@article{colless2018computation,
  title={Computation of molecular spectra on a quantum processor with an error-resilient algorithm},
  author={Colless, James I and Ramasesh, Vinay V and Dahlen, Dar and Blok, Machiel S and Kimchi-Schwartz, Mollie E and McClean, Jarrod R and Carter, Jonathan and de Jong, Wibe A and Siddiqi, Irfan},
  journal={Physical Review X},
  volume={8},
  number={1},
  pages={011021},
  year={2018},
  publisher={APS}
}

@article{herrman2022multi,
  title={Multi-angle quantum approximate optimization algorithm},
  author={Herrman, Rebekah and Lotshaw, Phillip C and Ostrowski, James and Humble, Travis S and Siopsis, George},
  journal={Scientific Reports},
  volume={12},
  number={1},
  pages={6781},
  year={2022},
  publisher={Nature Publishing Group UK London}
}

@article{magann2022feedback,
  title={Feedback-based quantum optimization},
  author={Magann, Alicia B and Rudinger, Kenneth M and Grace, Matthew D and Sarovar, Mohan},
  journal={Physical Review Letters},
  volume={129},
  number={25},
  pages={250502},
  year={2022},
  publisher={APS}
}

@article{zhu2022adaptive,
  title={Adaptive quantum approximate optimization algorithm for solving combinatorial problems on a quantum computer},
  author={Zhu, Linghua and Tang, Ho Lun and Barron, George S and Calderon-Vargas, Fernando A and Mayhall, Nicholas J and Barnes, Edwin and Economou, Sophia E},
  journal={Physical Review Research},
  volume={4},
  number={3},
  pages={033029},
  year={2022},
  publisher={APS}
}

@article{egger2021warm,
  title={Warm-starting quantum optimization},
  author={Egger, Daniel J and Mare{\v{c}}ek, Jakub and Woerner, Stefan},
  journal={Quantum},
  volume={5},
  pages={479},
  year={2021},
  publisher={Verein zur F{\"o}rderung des Open Access Publizierens in den Quantenwissenschaften}
}

@article{bravyi2020obstacles,
  title={Obstacles to variational quantum optimization from symmetry protection},
  author={Bravyi, Sergey and Kliesch, Alexander and Koenig, Robert and Tang, Eugene},
  journal={Physical review letters},
  volume={125},
  number={26},
  pages={260505},
  year={2020},
  publisher={APS}
}

@article{born1928beweis,
  title={Beweis des Adiabatensatzes},
  author={Born, Max and Fock, Vladimir},
  journal={Zeitschrift f{\"u}r Physik},
  volume={51},
  number={3},
  pages={165--180},
  year={1928},
  publisher={Springer}
}

@article{kato1950adiabatic,
  title={On the adiabatic theorem of quantum mechanics},
  author={Kato, Tosio},
  journal={Journal of the Physical Society of Japan},
  volume={5},
  number={6},
  pages={435--439},
  year={1950},
  publisher={The Physical Society of Japan}
}

@article{crooks2018performance,
  title={Performance of the quantum approximate optimization algorithm on the maximum cut problem},
  author={Crooks, Gavin E},
  journal={arXiv preprint arXiv:1811.08419},
  year={2018}
}

@article{harrigan2021quantum,
  title={Quantum approximate optimization of non-planar graph problems on a planar superconducting processor},
  author={Harrigan, Matthew P and Sung, Kevin J and Neeley, Matthew and Satzinger, Kevin J and Arute, Frank and Arya, Kunal and Atalaya, Juan and Bardin, Joseph C and Barends, Rami and Boixo, Sergio and others},
  journal={Nature Physics},
  volume={17},
  number={3},
  pages={332--336},
  year={2021},
  publisher={Nature Publishing Group UK London}
}

@article{wang2018quantum,
  title={Quantum approximate optimization algorithm for MaxCut: A fermionic view},
  author={Wang, Zhihui and Hadfield, Stuart and Jiang, Zhang and Rieffel, Eleanor G},
  journal={Physical Review A},
  volume={97},
  number={2},
  pages={022304},
  year={2018},
  publisher={APS}
}

@article{streif2020training,
  title={Training the quantum approximate optimization algorithm without access to a quantum processing unit},
  author={Streif, Michael and Leib, Martin},
  journal={Quantum Science \& Technology},
  volume={5},
  number={3},
  pages={034008},
  year={2020},
  publisher={IOP Publishing}
}

@article{pelofske2024scaling,
  title={Scaling whole-chip QAOA for higher-order Ising spin glass models on heavy-hex graphs},
  author={Pelofske, Elijah and B{\"a}rtschi, Andreas and Cincio, Lukasz and Golden, John and Eidenbenz, Stephan},
  journal={npj Quantum Information},
  volume={10},
  number={1},
  pages={109},
  year={2024},
  publisher={Nature Publishing Group UK London}
}

@article{maciejewski2026quantum,
  title={Quantum Approximate Optimization via Noise-Directed Adaptive Warm-Starting},
  author={Maciejewski, Filip B and Hadfield, Stuart and Wallis, Oscar and Pennington, George and Brandhofer, Sebastian and Woerner, Stefan and Egger, Daniel J and Venturelli, Davide},
  journal={arXiv preprint arXiv:2607.09368},
  year={2026}
}

@article{lee2023depth,
  title={A depth-progressive initialization strategy for quantum approximate optimization algorithm},
  author={Lee, Xinwei and Xie, Ningyi and Cai, Dongsheng and Saito, Yoshiyuki and Asai, Nobuyoshi},
  journal={Mathematics},
  volume={11},
  number={9},
  pages={2176},
  year={2023},
  publisher={MDPI}
}

@article{jansen2007bounds,
  title={Bounds for the adiabatic approximation with applications to quantum computation},
  author={Jansen, Sabine and Ruskai, Mary-Beth and Seiler, Ruedi},
  journal={Journal of Mathematical Physics},
  volume={48},
  number={10},
  year={2007},
  publisher={AIP Publishing}
}

@article{amin2009consistency,
  title={Consistency of the adiabatic theorem},
  author={Amin, Mohammad HS},
  journal={Physical review letters},
  volume={102},
  number={22},
  pages={220401},
  year={2009},
  publisher={APS}
}

@article{lbfgs,
  author = {Liu, Dong C. and Nocedal, Jorge},
  title  = {On the limited memory {BFGS} method for large scale optimization},
  journal = {Mathematical Programming},
  volume = {45}, number = {1}, pages = {503--528}, year = {1989}
}

@article{tong2005quantitative,
  title={Quantitative conditions do not guarantee the validity of the adiabatic approximation},
  author={Tong, DM and Singh, K and Kwek, Leong Chuan and Oh, Choo Hiap},
  journal={Physical review letters},
  volume={95},
  number={11},
  pages={110407},
  year={2005},
  publisher={APS}
}

@article{spsa,
  author = {Spall, James C.},
  title  = {Multivariate stochastic approximation using a simultaneous perturbation gradient approximation},
  journal = {IEEE Transactions on Automatic Control},
  volume = {37}, number = {3}, pages = {332--341}, year = {1992}
}

@book{bransden2000quantum,
  author    = {Bransden, B. H. and Joachain, C. J.},
  title     = {Quantum Mechanics},
  edition   = {2},
  publisher = {Pearson Education},
  address   = {Harlow},
  year      = {2000},
  isbn      = {9780582356917}
}

@book{messiah1962quantum,
  author      = {Messiah, Albert},
  title       = {Quantum Mechanics},
  volume      = {2},
  publisher   = {North-Holland},
  address     = {Amsterdam},
  year        = {1962},
  translator  = {Temmer, G. M.}
}

\onecolumngrid
\appendix
\section{Light Cone of Errors}

\label{appendix:lightcone}

Propagating the error to the end of circuit via unitary conjugation gives an error which acts as a unitary applied to the end of the circuit. Investigating the structure of this unitary gives us insight into error propagation. 
\begin{align}
    \tilde{U} &=K_j \hat{U}_p\ldots \hat{U}_{s+1} U_{s} \ldots U_{1}\\
    &=\hat{U}U_p\ldots U_1\\
    &=\hat{U}U.
\end{align}
Defining $\mathcal{N}_p(j)$ as nodes in the $p$ size light cone of $j$ (e.g. in \eqref{big} $\mathcal{N}(j) = \mathcal{N}_1(j)$) and $\mathcal{L}_p(j)$ as the edges in the $p$ size light cone of $j$ where $d(j,v)$ is the shortest distance from node j to v 
\begin{align}
    \mathcal{N}_p(j)&= \{v\in V: d(j,v)\leq p\},\\
    \mathcal{L}_p(j)&=\{(u,v) \in E: u \wedge v \in \mathcal{N}_p(j)\},
\end{align}
then, for $K_j = Z_j$
\begin{align}
    \hat{U} = &Z_j 
    \bigg[
    \prod_{l=s+2}^{\overset{\curvearrowleft}{p}}
    e^{i\beta_l\sum_{i\in\mathcal{N}_{l-s-1}(j),i\neq j}X_i}e^{-i\beta_lX_j}
    e^{i\gamma_l\sum_{(i,k)\in\mathcal{L}_{l-s-1}(j)}Z_iZ_k}
    \bigg] \\
    &\times\  e^{-2i\beta_{s+1}X_j} 
    \bigg[
    \prod_{l=s+2}^{\overset{\curvearrowright}{p}}
    e^{-i\gamma_l\sum_{(i,k)\in\mathcal{L}_{l-s-1}(j)}Z_iZ_k}
    e^{-i\beta_l\sum_{i\in\mathcal{N}_{l-s-1}(j)}X_i}
    \bigg],
\end{align}
For $K_j = X_j$
\begin{align}
    \hat{U} = &X_j 
    \bigg[
    \prod_{l=s+2}^{\overset{\curvearrowleft}{p}}
    e^{i\beta_l\sum_{i\in\mathcal{N}_{l-s}(j)}X_i}
    e^{i\gamma_l\sum_{(i,k)\in\mathcal{L}_{l-s}(j),i,k\notin \mathcal{N}(j)}Z_iZ_k}
    e^{-i\gamma_l\sum_{i\in\mathcal{N}(j)}Z_iZ_j}
    \bigg] \\
    &\times\  e^{i\beta_{s+1}\sum_{i\in\mathcal{N}_1(j)}X_i} 
    e^{-2i\gamma_{s+1}\sum_{i\in\mathcal{N}(j)}Z_iZ_j} 
    e^{-i\beta_{s+1}\sum_{i\in\mathcal{N}_1(j)}X_i} \\
    &\times\bigg[\prod_{l=s+2}^{\overset{\curvearrowright}{p}}
    e^{-i\gamma_l\sum_{(i,k)\in\mathcal{L}_{l-s}(j)}Z_iZ_k}
    e^{-i\beta_l\sum_{i\in\mathcal{N}_{l-s}(j)}X_i}
    \bigg].
\end{align}
This acts on the light cone of the affected qubit dependent on the kind of error. For $K_j=Z$ this acts on qubits in $\mathcal{N}_{p-s-1}(j)$ and depends on $2(p-s)-1$ parameters. For $K_j=X$ this acts on qubits in $\mathcal{N}_{p-s}(j)$ and depends on $2(p-s)+1$ parameters. 
Thus the impact of the error is based on the graph structure and connectivity of the affected qubits.

\section{Bounds on Error Resilience}
\subsection{Proof of Theorem \ref{theorem}}\label{appendix:theorem1}
\begin{proof}
Considering firstly the case where $K=Z$. Using \eqref{dragZ} and the fact that we measure in the computational basis, once propagated through the circuit, the $Z$ has no effect on the output distribution. Therefore,
\begin{align}
    \mathcal{D}(\tilde U_{j,s}^Z,U) &= \mathcal{D}(Z_j \hat{U}_p \ldots \hat{U}_{s+1} U_s \ldots U_1, U_p\ldots U_{1})\\
    &= \mathcal{D}(\hat{U}_p \ldots \hat{U}_{s+1} U_s \ldots U_1, U_p\ldots U_{1}).\label{dragZD}
\end{align}

We bound the TVD by the operator norm via the trace distance using the below lemma. 
\begin{lemma}\label{lemma}
The TVD defined in \eqref{tvd} can be bounded by the operator norm,
\begin{align}
    \mathcal{D}(U,V) \leq \|U-V\|.
\end{align}
where 
\begin{align}
    \mathcal{D}(\tilde{U},U) =  \| |\tilde U \ket{+^n}|^2 - | U \ket{+^n}|^2 \|_\text{tvd}.
\end{align}
and $\|\cdots\|$ denotes the operator norm. Proof in Appendix \ref{appendix:lemma1}
\end{lemma}
Using this, \eqref{dragZD} is bounded by
\begin{align}
    \mathcal{D}(\hat{U}_p \ldots \hat{U}_{s+1} U_s \ldots U_1, U_p\ldots U_{1}) \leq \|\hat{U}_p \ldots \hat{U}_{s+1} U_s \ldots U_1 -  U_p\ldots U_{1}\|.
\end{align}
Using the invariance of the operator norm under unitaries 
\begin{align}
    \|\hat{U}_p \ldots \hat{U}_{s+1} U_s \ldots U_1 -  U_p\ldots U_{1}\| = \|\hat{U}_p \ldots \hat{U}_{s+1} -  U_p\ldots U_{s+1}\|.
\end{align}
Then using the triangle inequality on the operator norm
\begin{align}
    \mathcal{D}(\tilde U_{j,s}^Z,U) 
    &\leq \|\hat{U}_p \ldots \hat{U}_{s+1} - \hat{U}_p \ldots \hat{U}_{s+2} U_{s+1}\| 
    \\&\quad+ \|\hat{U}_p \ldots \hat{U}_{s+2}U_{s+1}- \hat{U}_p \ldots \hat{U}_{s+3} U_{s+2}U_{s+1}\|\notag\\&\quad\dots\notag\\
    &\quad+\|\hat{U}_p\hat{U}_{p-1}U_{p-2}\dots U_{s+1}- \hat{U}_p U_{p-1}\dots U_{s+1}\|\notag\\
    &\quad+\|\hat{U}_pU_{p-1}\dots U_{s+1}- U_p\dots U_{s+1}\|\notag\\
    & = \sum_{i = s+1}^{p} \|\hat{U_i}- U_i\|.
\end{align}
where in the final line we can again use the unitary invariance of the operator norm to remove the unitaries that appear in both terms of each norm

Using \eqref{big} we calculate $\|\hat{U_i}- U_i\|$ 
\begin{align}
    \| \hat{U}_i - U_i \| &= \| e^{-i \beta_i X_j} - e^{i \beta_i X_j} \|\\
&=2|\sin(\beta_i)|\\
&\leq 2|\beta_i| \label{Z}.
\end{align}
Therefore
\begin{align}
    \mathcal{D}(\tilde{U}^Z_{j,s},U) \leq 2\sum_{i = s+1}^{p}|\sin{\beta_i}|\leq2\sum_{i = s+1}^{p}|\beta_i|.
\end{align}

For $K = X$, a similar proof follows. Using \eqref{dragX} and the fact that we begin in $\ket{+^n}$, the $X$ error propagated to the beginning has no effect on the output distribution as so can be dropped from the TVD 
\begin{align}
    \mathcal{D}(\tilde U^X_{j,s},U) &= \mathcal{D}(U_p \ldots U_{s+1} \hat{U}_s \ldots \hat{U}_1 X_j, U_p\ldots U_{1})\\
    &= \mathcal{D}(U_p \ldots U_{s+1} \hat{U}_s \ldots \hat{U}_1, U_p\ldots U_{1}).
\end{align}
Following the same steps as the $K=Z$ proof we get that 
\begin{align}
    \mathcal{D}(\tilde{U}_{j,s}^X,U) 
    &\leq 
    \sum_{i = 1}^{s} \|\hat{U}_i- U_i\|.
\end{align}
Using \eqref{big}
\begin{align}
    \| \hat{U}_i - U_i\| &= \| e^{-i\gamma_i\sum_{k\sim j}Z_{k}Z_j-}e^{i\gamma_i\sum_{k\sim j}Z_{k}Z_j}\|\\
    &= \max_x |e^{-if(x)} - e^{if(x)}|\\
    &= \max_x |2\sin{(f(x))}|,
\end{align}
where 
\begin{align}
    f(x) = \gamma_i(-1)^{x_j}\sum_{k\in\mathcal{N}(j)}(-1)^{x_k},
\end{align}
it then follows that, for node $j$ with degree $d_j$ that $-d_j\gamma_i \leq f(x) \leq d_j\gamma_i$, meaning
\begin{align}
    \| \hat{U}_i - U_i\| &= \max_{\omega \in \{0,\ldots,d_j\}}2|\sin{(\omega\gamma_i})|\\
    &\leq2d_j|\gamma_i|\label{X},
\end{align}
and we have that
\begin{align}
    \mathcal{D}(\tilde{U}^X_{j,s},U) \leq 2\max_{\omega \in \{0,\ldots,d_j\}}\sum_{i = 1}^{s}|\sin{\omega\gamma_i}|\leq2d_j\sum_{i = 1}^{s}|\gamma_i|.
\end{align}

In the case $K=Y$ we can represent a $Y$ error as $-iZX$ and can drag the $X$ and $Z$ through individually and drop the constant phase with the $Z$. 

\begin{align}
    \mathcal{D}(\tilde{U},U) &= \mathcal{D}(U_p\ldots U_{s+1}Y_{j}U_{s}\ldots U_1,U_p\ldots U_1)\\
    &= \mathcal{D}(U_p\ldots U_{s+1}(-iZ_{j}X_{j})U_{s}\ldots U_1,U_p\ldots U_1)\\
    &=\mathcal{D}((-i)Z_{j}\hat U_p^Z(j)\ldots \hat U_{s +1}^Z(j)\hat U_{s}^X(j)\ldots \hat U_1^X(j) X_{j},U_p\ldots U_1)\\
    &=\mathcal{D}(\hat U_p^Z(j)\ldots \hat U_{s +1}^Z(j)\hat U_{s}^X(j)\ldots \hat U_1^X(j),U_p\ldots U_1),
\end{align}
where here $\hat U^K_i$ denotes the $i$th layer having had a Pauli $K$ propagated through. 

By a similar use of the triangle inequality, this case is just additive:
\begin{align}
    \mathcal{D}(\tilde{U}^Y_{j,s},U)\leq \sum_{i=1}^{s}\|\hat U_i ^X - U_i\|+\sum_{i=s+1}^{p}\|\hat U_i ^Z - U_i\|.
\end{align}

Grouping together terms we get 
\begin{align}
\mathcal{D}(\tilde{U}^K_{j,s},U)
  \leq 2\Big[\,
    &\chi_{K,X}\sum_{i=s+1}^{p}|\sin\beta_i| \\
  + &\chi_{K,Z}\max_{\omega\in\{0,\dots,d_j\}}\sum_{i=1}^{s}|\sin(\omega\gamma_i)|
  \,\Big] \\
  \leq 2\Big[\,
    &\chi_{K,X}\sum_{i=s+1}^{p}|\beta_i|
  + \chi_{K,Z}\,d_j\sum_{i=1}^{s}|\gamma_i|
  \,\Big].
\end{align}
\end{proof}

\subsection{Proof of Lemma 
\ref{lemma}}\label{appendix:lemma1}
\begin{proof}
For two arbitrary unitary quantum circuits U and V, the TVD between the final distributions after measuring in the computational basis is \begin{align}
\mathcal{D}(U\ket{\psi},V\ket{\psi}) = \| |U \ket{\psi}|^2 - | V \ket{\psi}|^2 \|_\text{tvd} = \frac{1}{2} \sum_{x\in\{0,1\}^n}\left||\bra{x}U\ket{\psi}|^2-|\bra{x}V\ket{\psi}|^2\right|,
\end{align}and trace distance between the quantum states at the end of the circuit
\begin{align}
\mathcal{T}(U\ket{\psi},V\ket{\psi}) =  \frac{1}{2} \|U\ket{\psi}\bra{\psi}U^\dagger - V\ket{\psi}\bra{\psi}V^\dagger \|_1 = \frac{1}{2}\text{Tr}|U\ket{\psi}\bra{\psi}U^\dagger - V\ket{\psi}\bra{\psi}V^\dagger|.
\end{align}
The TVD is bounded by the trace distance \cite{nielsen2010quantum}
\begin{align}
    \mathcal{D}(U\ket{\psi},V\ket{\psi})\leq\mathcal{T}(U\ket{\psi},V\ket{\psi}).
\end{align}
The trace distance is bounded by the operator norm:
\begin{align}
    \mathcal{T}(U\ket{\psi}, V\ket{\psi}) = 
    \sqrt{1-|\bra{\psi}U^\dagger V\ket{\psi}|^2} \leq \|(U-V)\ket{\psi}\| \leq \|U-V\|.
\end{align}
Therefore we get that
\begin{align}
\mathcal{D}(U\ket{\psi},V\ket{\psi})\leq\mathcal{T}(U\ket{\psi},V\ket{\psi})\leq\|U-V\|.
\end{align}
\end{proof}

\section{Error Resilience Without Triviality}\label{appendix:performance}

The bound \eqref{DK} derived for error resilience shows that smaller angles decrease the effect of errors on the output state. This may seem trivial as smaller angles mean the circuit does not get as far from the initial state and hence errors have less of an effect. However, we show there exists a separation between the regime of error resilience and the regime of a trivial circuit that does not move far from the initial state, via similar methods to the proof of \eqref{DK}. 

\begin{theorem} \label{theorem:performance}
Let $C(x)$ be a cost function on graph $\mathcal{G}$ with edge set $E$ and $n$ nodes, where node $i$ has degree $d_i$. 

We consider a QAOA circuit $U$ with $p$ layers and parameters $\gamma= \{\gamma_i\}, \beta = \{\beta_i\}$ for $i\in\{1,\ldots,p\}$. 

The TVD between the initial state and the final state after the circuit is
\begin{align}
    \mathcal{D}(U,I)\leq2\sum_{i=1}^p \left[\max_{x\in\{0,1\}^n}|\sin(\frac{\gamma_iC(x)}2)| + \max_{k\in\{-n,\ldots,n\}}|\sin(\frac{\beta_ik}2)| \right].
\end{align}
\end{theorem}
\begin{proof}
    Looking at the TVD between the initial and final states, using Lemma \ref{lemma} and expanding the circuit into it's components
\begin{align}
    D(U,I) &\leq \| e^{i\beta_pM}e^{i\gamma_pC}\ldots e^{i\beta_1M}e^{i\gamma_1C}\ket{+^n} - \ket{+^n}\| \\
    &\leq\|e^{i\gamma_1C}\ket{+^n}-\ket{+^n}\| + \| e^{i\gamma_1C} e^{i\beta_1M} \ket{+^n} - e^{i\gamma_1C}\ket{+^n}\|\\
    &\quad +\ldots\\&\quad+\|e^{i\beta_pM}e^{i\gamma_pC}\ldots e^{i\beta_1M}e^{i\gamma_1C}\ket{+^n} - e^{i\gamma_pC}\ldots e^{i\beta_1M}e^{i\gamma_1C}\ket{+^n}\|\\
    &\leq\|e^{i\gamma_1C}-I\| + \|e^{i\beta_1M}-I\| + \ldots + \|e^{i\gamma_pC}-I\| + \|e^{i\beta_pM}-I\|\\
    &=\sum_{j=1}^{p} \|e^{i\gamma_jC}-I\| + \|e^{i\beta_jM}-I\|.
\end{align}
These individually are 
\begin{align}
    \|e^{i\gamma_sC}-I\| = \max_{x\in\{0,1\}^n}|e^{i\gamma_sC(x)}-1| =  \max_{x\in\{0,1\}^n}2|\sin(\frac{\gamma_sC(x)}2)|,
\end{align}
and
\begin{align}
    \|e^{i\beta_sM}-I\| = \max_{k\in\{-n,\ldots,n\}}|e^{i\beta_sk}-1| =  \max_{k\in\{-n,\ldots,n\}}2|\sin(\frac{\beta_sk}2)|.
\end{align}
plugging these back in we find 
\begin{align}
    D(U,I)\leq2\sum_{i=1}^p \left[\max_{x\in\{0,1\}^n}|\sin(\frac{\gamma_iC(x)}2)| + \max_{k\in\{-n,\ldots,n\}}|\sin(\frac{\beta_ik}2)| \right].
\end{align}
\end{proof}
This gives an indication of the ability of the circuit to move the state away from the initial state i.e. it's ability to do anything at all. The goal would be to minimize the error induced distance \eqref{DK} while ensuring this distance is not too small.

If we take  
\begin{align}
    \mathcal{D}(\tilde U, U) \leq \delta ,
\end{align}
where $\delta\ll1$ means the circuit is robust to noise, we bound $\Gamma=\sum_i^p\gamma_i$ and $B=\sum_i^p\beta_i$ to this regime. 

If we choose 
\begin{align}
    \Gamma \sim\frac{\delta}{2d_j} \quad \text{and} \quad B\sim\frac{\delta}{2},
\end{align}
then
\begin{equation}
  D(\tilde U, U)\leq \delta \quad\text{and} \qquad
  D(U,I) \lesssim \frac{|E|}{2d_j}\delta+ \frac{n}{2}\delta \sim n\delta \gg \delta.
\end{equation}
This indicates that we would expect to see a separation between this distance and the error distances in this regime though this upper bound doesn't guarantee a gap.

\begin{figure}
    \centering
    \includegraphics[width=0.6\linewidth]{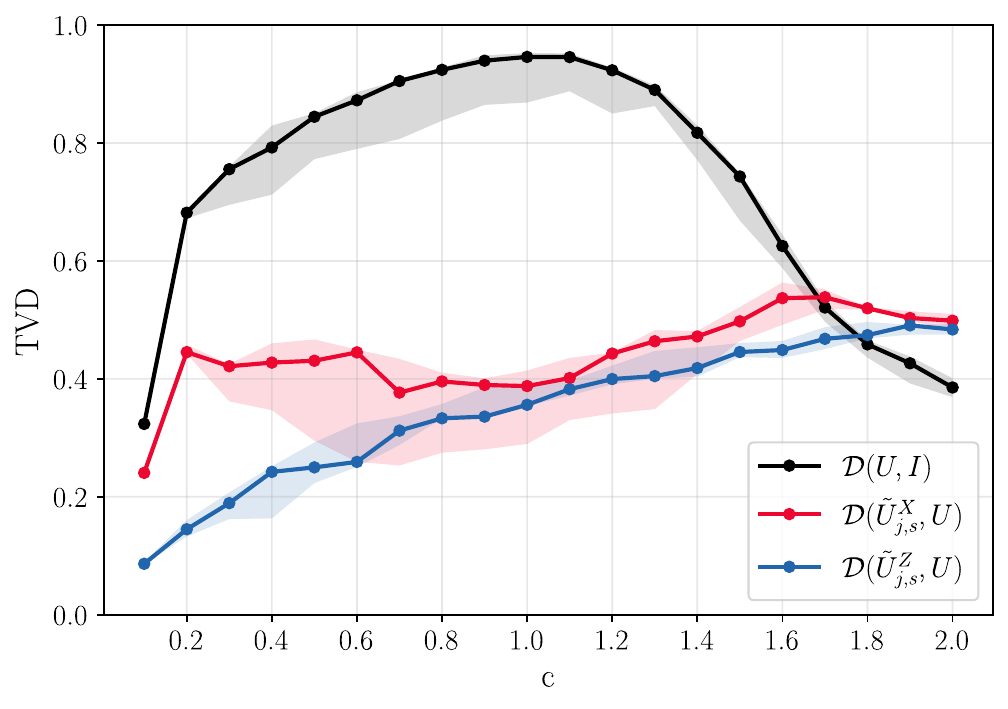}
    \caption{Median total variation distance for QAOA with $p = 10$ across five 10 node 3-regular graphs with predicted parameters generated using \textit{balanced method} from \cite{ieva} scaled by a factor $c$ i.e. $\{\gamma_i,\beta_i\}\rightarrow\{c\gamma_i,c\beta_i\}$. We plot the distance $\mathcal{D}(U,I)$ to show how far the set of parameters can move the state from the initial state. From a scaling of roughly $c=0.2$ we can see the algorithm is able to move from the initial state
    We also plot the error distances $\mathcal{D}(\tilde{U}^K_{j,s},U)$ for $K\in\{X,Y\}$. The median taken across all $j$ and $s$ before the median is taken across graphs. This shows regimes of small parameter scaling where the error distances are not large but the circuit is able to meaningfully move the state from the initial state. The shaded regions show the interquartile range. }
    \label{fig:scaling}
\end{figure}
To validate that there can indeed be a gap in practice, we show an example in Fig.~\ref{fig:scaling}. For this test we ran QAOA with predicted parameters scaled by a varying factor $c$. The results show that for small parameters, the circuit is still able to move meaningfully far from the initial state. In the small angle regime, it is possible to have a large distance between the resultant state and initial state whilst maintaining a small average distance in the presence of errors, showing the bound in \eqref{DK} is not trivial in that it is not just small for small parameters because the state is not moving. The error distances increase with $c$ as expected.

\section{Adiabatic Unstructured Quantum Search with Two Schedules}\label{appendix:Adiabatic}

Here we derive the condition for local adiabatic evolution for unstructured quantum search for the two schedule adiabatic algorithm as given in \eqref{ham}-\eqref{H0}. This proof closely follows that of Roland and Cerf in \cite{roland2002quantum} as well as that of Isermann in \cite{isermann2021optimal} but with various parts generalized to two schedules.

\subsection{Finding $g(t)$} \label{appendix:Adiabaticgap}
To find the gap $g(t)$ between the two lowest energies of the quantum search Hamiltonian we must find the eigenvalues. We expand $\ket{+^n}$ in terms of the marked and unmarked states and define the state $\ket{r}$ as the superposition of unmarked states
\begin{align}
    \ket{+^n} &= \frac{1}{\sqrt{N}}\sum_x\ket{x} =a\ket{m} + b\ket{r}\\
    \ket{r} &= \frac{1}{\sqrt{N-1}}\sum_{x\neq m}\ket{x},
\end{align}
so that we can write $H$ in terms of the orthonormal basis $\{\ket{m},\ket{r}\}$. Using
\begin{align}
    a&= \braket{m|+^n} = \frac{1}{\sqrt{N}}\\
    b&= \sqrt{\frac{N-1}{N}}=\sqrt{1-a^2},
\end{align}
we rewrite $H$ in terms of its components $H_{ij}$
\begin{align}
    H_{mm} &= \bra{m}H\ket{m} \\&= 
    \gamma(t) \bra{m}(I-\ket{m}\bra{m})\ket{m} + 
    \beta (t) \bra{m}(I-\ket{+^n}\bra{+^n})\ket{m}\\
    &=\beta(t)\braket{m|m} - \beta(t) |\braket{m|+^n}|^2\\
    &=\beta(t)-\frac{\beta(t)}{N}.
\end{align}
Similarly 
\begin{align}
    H_{mr} = H_{rm} &=-\beta(t)\frac{\sqrt{N-1}}{N},\\
    H_{rr} &=\gamma(t)+\frac{\beta(t)}{N}.
\end{align}
This gives
\begin{align}
H=
\begin{pmatrix}
    \beta(t)-\frac{\beta(t)}{N} & -\beta(t)\frac{\sqrt{N-1}}{N}\\
    -\beta(t)\frac{\sqrt{N-1}}{N} & \gamma(t)+\frac{\beta(t)}{N}
\end{pmatrix}.
\end{align}
We find the eigenvalues $E_{0/1}$ of this matrix 
\begin{align}\label{energies}
    E_{1/0}(t) = \frac{1}{2}\bigg(
    \beta(t) + \gamma(t) \pm \sqrt{(\beta(t)+\gamma(t))^2-4\beta(t)\gamma(t)\left(1-\frac{1}{N}\right)}
    \bigg),
\end{align}
so that the gap is given by
\begin{align}\label{g}
    g(t) &= E_1(t)-E_0 (t) \\&=\sqrt{(\beta(t)+\gamma(t))^2-4\beta(t)\gamma(t)\left(1-\frac{1}{N}\right)} \\&= \sqrt{(\beta(t)-\gamma(t))^2+\frac{4\beta(t)\gamma(t)}{N}}.
\end{align}

\subsection{Finding $|\bra{E_1;t}\frac{dH}{dt}\ket{E_0;t}| $}\label{appendix:AdiabaticdH}
 
To find $|\bra{E_1;t}\frac{dH}{dt}\ket{E_0;t}| $, we use an ansatz for the eigenstates of $H(t)$, $\ket{E_i(\gamma,\beta)}$ in the basis $\ket{+^n}$ and $\ket{m}$,
\begin{align}
\ket{E_i(\gamma,\beta)}=\mu_i(\gamma,\beta)\ket{\psi_0}+\nu_i(\gamma,\beta)\ket{m},
\end{align}
where $\mu_i,\nu_i$ are real. Substituting these ansatz into the matrix element gives
\begin{align}
    |\bra{E_1;t}\frac{dH}{dt}\ket{E_0;t}| &= |\bra{E_1;t}\bigg[\frac{d\gamma}{dt}\frac{\partial H}{\partial \gamma} + \frac{d\beta}{dt}\frac{\partial H}{\partial \beta}\bigg]\ket{E_0;t}|\\
    &=|\frac{d\gamma}{dt}\bra{E_1} H_m\ket{E_0} + \frac{d\beta}{dt}\bra{E_1} H_0\ket{E_0}|\\
    &=|\frac{d\gamma}{dt}(\mu_1\mu_0\left(1-\frac{1}{N}\right)) + \frac{d\beta}{dt}(\nu_1\nu_0\left(1-\frac{1}{N}\right))|.\label{dhdt}
\end{align}

Using orthonormality of the eigenstates we get the following relations
\begin{align}
    \braket{E_0|E_1} &=
    \mu_0\mu_1+\nu_0\nu_1+\frac{1}{\sqrt{N}}(\mu_0\nu_1+\nu_0\mu_1)=0,\\
    \braket{E_i|E_i} &=
    \mu_i^2+\nu_i^2+\frac{2}{\sqrt{N}}\mu_i\nu_i=1\label{orthogonal}.
\end{align}
Using the fact that these are eigenstates, 
\begin{align}
H(\gamma,\beta)\ket{E_i(\gamma,\beta)}&=E_i\ket{E_i(\gamma,\beta)},
\end{align}
gives
\begin{align}
    (\gamma(I-\ket{m}\bra{m}) +\beta(I-\ket{\psi_0}\bra{\psi_0}))(\mu_i\ket{\psi_0} + \nu_i\ket{m})&= E_i(\mu_i\ket{\psi_0} + \nu_i\ket{m}).
\end{align}
Expanding and grouping like terms gives 
\begin{align}
    [\mu_i\gamma-\frac{\nu_i\beta}{\sqrt{N}}-E_i\mu_i]\ket{\psi_0}+
    [\nu_i\beta-\frac{\mu_i\gamma}{\sqrt{N}}-E_i\nu_i]\ket{m} =0.
\end{align}
Since $\ket{\psi_0}$ and $\ket{m}$ are linearly independent, each coefficient must vanish
independently, giving the following relations
\begin{align}
    \nu_i(\beta-E_i)&=\frac{\mu_i\gamma}{\sqrt{N}}\label{alpha2}\\
    \mu_i(\gamma-E_i)&=\frac{\nu_i\beta}{\sqrt{N}}\label{beta2}.
\end{align}

Rearranging and squaring \eqref{alpha2} gives
\begin{align}
    \nu_i^2=\mu_i^2\frac{\gamma^2}{N(\beta-E_i)^2},
\end{align}
substituting that into the orthogonality condition \eqref{orthogonal}, we get
\begin{align}
    \mu_i^2\bigg(
    1+\frac{\gamma^2}{N(\beta-E_i)^2} + \frac{2\gamma}{N(\beta-E_i)}
    \bigg)=1,
\end{align}
which gives
\begin{align}\label{mu}
    \mu_i^2 = \frac{N(\beta-E_i)^2}{(N-1)(\beta-E_i)^2+(\gamma+\beta-E_i)^2}
\end{align}
and 
\begin{align}\label{nu}
    \nu_i^2 = \frac{\gamma^2}{(N-1)(\beta-E_i)^2+(\gamma+\beta-E_i)^2}.
\end{align}

Using the energies $E_0$ and $E_1$,we calculated previously \eqref{energies} and relations between them
\begin{align}
    g(t) &= E_1(t)-E_0(t)\label{eq:g}\\
    E_0+E_1 &= \gamma+\beta\\
    E_0E_1&=\beta\gamma\left(1-\frac{1}{N}\right),
\end{align}
which give,
\begin{align}
  (\beta-E_0)(\beta-E_1)=-\frac{\beta\gamma}{N},\label{eq:consequences2}\\
  \gamma+\beta-E_i=E_{1-i}.
  \label{eq:consequences}
\end{align}
we evaluate \eqref{dhdt}.

Using the definition for $\mu_i$ in \eqref{mu}
\begin{align}
    \mu_1\mu_0\left(1-\frac{1}{N}\right) &= \left(1-\frac{1}{N}\right) 
    \sqrt{\frac{N(\beta-E_1)^2}{(N-1)(\beta-E_1)^2+(\gamma+\beta-E_1)^2}}
    \\&\qquad\qquad\qquad\qquad\times
    \sqrt{\frac{N(\beta-E_0)^2}{(N-1)(\beta-E_0)^2+(\gamma+\beta-E_0)^2}}.
\end{align}

Using \eqref{eq:consequences2} the numerator simplifies to 
\begin{equation}
  N^2(\beta-E_0)^2(\beta-E_1)^2
  =\left(N(\beta-E_0)(\beta-E_1)\right)^2
  =\beta^2\gamma^2.
  \label{eq:num}
\end{equation}

Using \eqref{eq:g} and \eqref{eq:consequences}, the denominator simplifies to 
\begin{equation}
  \bigl((N-1)(\beta-E_1)^2+E_0^2\bigr)\bigl((N-1)(\beta-E_0)^2+E_1^2\bigr) = (N-1)\gamma^2g^2,
\end{equation}
such that the expression becomes 
\begin{align}
    \mu_1\mu_0\left(1-\frac{1}{N}\right)
 &= \left(\frac{N-1}{N}\right)\frac{\sqrt{\gamma^2\beta^2}}{\sqrt{(N-1)\gamma^2g^2}}\\&=\frac{\sqrt{N-1}}{N}\frac{\beta}{g}.
\end{align}

If we fix $\mu_0,\mu_1\geq0$ then, via \eqref{alpha2}, we know $\nu_0\geq0$ and $\nu_1\leq0$. 

Using the definition for $\nu_i$ in \eqref{nu}
\begin{align}
    \nu_1\nu_0\left(1-\frac{1}{N}\right) &= \left(1-\frac{1}{N}\right)\times - 
    \sqrt{\frac{\gamma^2}{(N-1)(\beta-E_1)^2+(\gamma+\beta-E_1)^2}}
    \\&\qquad\qquad\times
    \sqrt{\frac{\gamma^2}{(N-1)(\beta-E_0)^2+(\gamma+\beta-E_0)^2}}\\
    &=-\left(1-\frac{1}{N}\right)\frac{\gamma^2}{\sqrt{(N-1)\gamma^2g^2}}\\&=-\frac{\sqrt{N-1}}{N}\frac{\gamma}{g}.
\end{align}
Finally putting this back into \eqref{dhdt} 
\begin{align}
    |\bra{E_1;t}\frac{dH}{dt}\ket{E_0;t}| &= |\frac{d\gamma}{dt}(\mu_1\mu_0\left(1-\frac{1}{N}\right)) + \frac{d\beta}{dt}(\nu_1\nu_0\left(1-\frac{1}{N}\right))|\\
    &=\frac{\sqrt{N-1}}{Ng}|\frac{d\gamma}{dt}\beta - \frac{d\beta}{dt}\gamma|.
\end{align}

\subsection{Proof that \eqref{rate} holds in the crossing region}\label{appendix:rate_crossing}

To see that \eqref{rate} holds for all $\gamma,\beta$ in this region, write
$\gamma = (1-\varepsilon_\gamma)\delta$ and $\beta = (1+\varepsilon_\beta)\delta$, where
\begin{align}
    |\gamma-\beta| = |(\varepsilon_\gamma + \varepsilon_\beta)\delta|&< \frac{2\delta\sqrt{(1+\varepsilon_\beta)(1-\varepsilon_\gamma)}}{\sqrt{N}},\\
    |\varepsilon_\gamma + \varepsilon_\beta|&<\frac{2\sqrt{(1+\varepsilon_\beta)(1-\varepsilon_\gamma)}}{\sqrt{N}},\\
    &\approx \frac{2}{\sqrt{N}}.\label{epsilons}
\end{align}
Inserting this $\gamma$ and $\beta$ into the left hand side of \eqref{constraint}, we get
\begin{align}
    \frac{d\gamma}{dt}\beta - \frac{d\beta}{dt}\gamma
    &= \delta\left(\frac{d\gamma}{dt} - \frac{d\beta}{dt}\right)
    + \delta\left(\frac{d\gamma}{dt}\varepsilon_\beta + \frac{d\beta}{dt}\varepsilon_\gamma\right).
\end{align}
and lower bound this using the worst case, $|\varepsilon_\gamma|,|\varepsilon_\beta| \leq 2/\sqrt{N}$  by monotonicity $\dot\gamma \geq 0$ and $\dot\beta \leq 0$, 
\begin{align}
    \delta\,(\frac{d\gamma}{dt} - \frac{d\beta}{dt}) + \delta(\frac{d\gamma}{dt}\varepsilon_\beta + \frac{d\beta}{dt}\varepsilon_\gamma)&\geq\delta\,(\frac{d\gamma}{dt} - \frac{d\beta}{dt}) + \delta(\frac{d\gamma}{dt}\left(-\frac{2}{\sqrt{N}}\right)+\frac{d\beta}{dt}\left(\frac{2}{\sqrt{N}}\right))
    \\&=\delta\,(\frac{d\gamma}{dt} - \frac{d\beta}{dt}) - \frac{2\delta}{\sqrt N}(\frac{d\gamma}{dt} - \frac{d\beta}{dt}).
\end{align}
So for large $N$, 
\begin{align}
 \delta(1 - \frac{2}{\sqrt N})(\frac{d\gamma}{dt} - \frac{d\beta}{dt}) &\leq \frac{8\eta(1-\varepsilon_\gamma)^{3/2}(1+\varepsilon_\beta)^{3/2}\delta^3}{N},\\
 &\approx \frac{8\eta\delta^3}{N} \left(1+\mathcal{O}(\tfrac{1}{\sqrt{N}}) \right),\\
 &\lesssim \frac{8\eta\delta^3}{N},\\
    (\frac{d\gamma}{dt} - \frac{d\beta}{dt}) &\leq \frac{8\eta\delta^2}{N}\left(1+\frac{2}{\sqrt{N}-2}\right),\\
    &\approx\frac{8\eta\delta^2}{N}.
\end{align}
giving \eqref{rate} for the whole region.

\newpage
\section{Noise-informed Parameter Optimization} \label{appendix:penalty}

Fig.~\ref{fig:penalty-plot-avg} in the main text shows the noisy approximation ratios for different parameter optimization methods averaged over 5 graph instances. For completeness the plots in Fig.~\ref{fig:All_graphs} show the individual results for all graphs. The separate results show the methods outlined in section~\ref{section:penalty} are perhaps more or less effective for certain types of graphs though the results are inconclusive. 

\begin{figure}[b]
    \centering
    \includegraphics[width=0.78\linewidth]{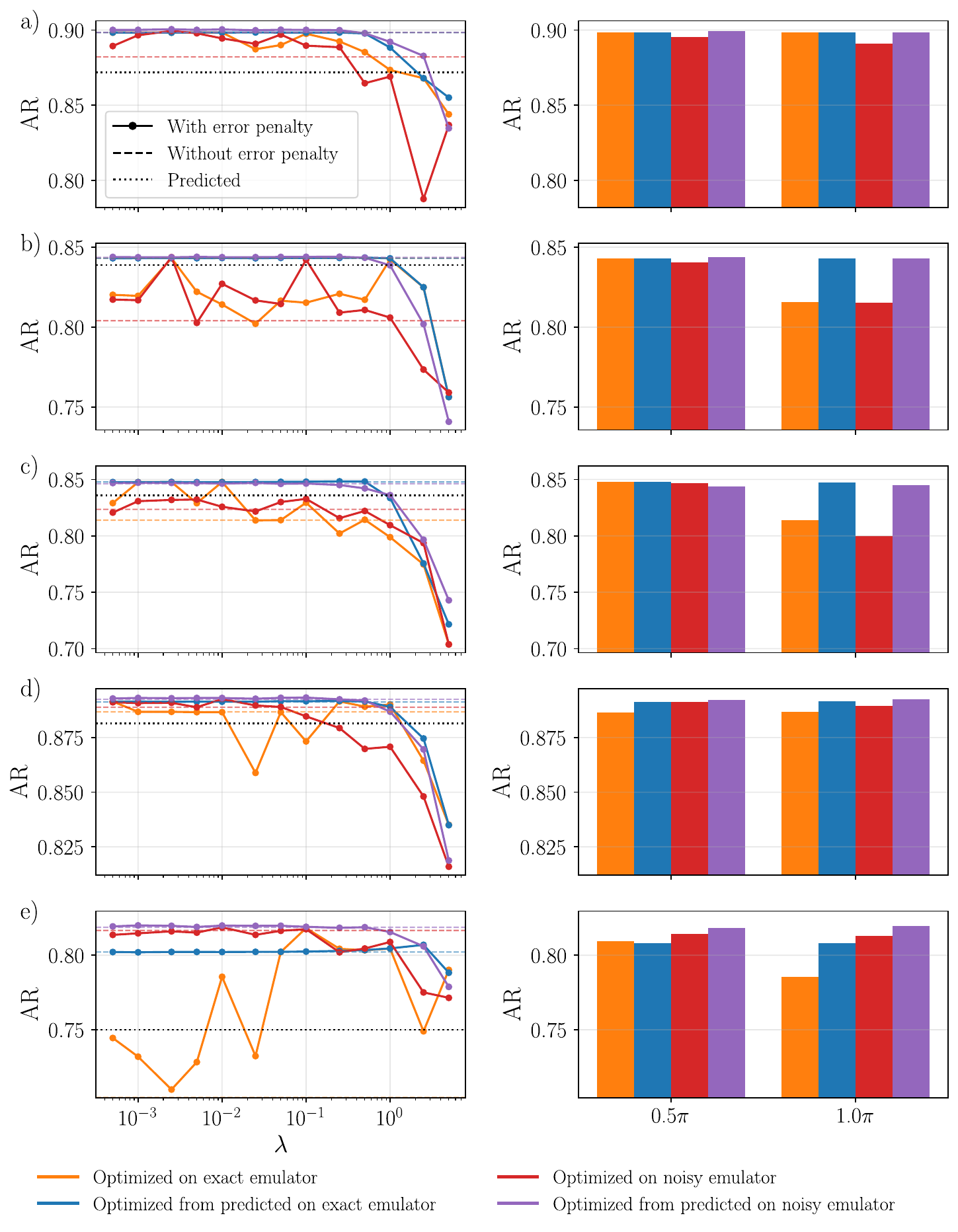}
    \caption{Noisy approximation ratios for the individual graphs averaged over in Fig.~\ref{fig:penalty-plot-avg}. All graphs had 6 nodes, a) and b) show results for Erd\H{o}s -- R\'enyi graphs, c) for the line graph, and d) and e) for 3-regular graphs. As in Fig.~\ref{fig:penalty-plot-avg} the plots on the left-hand side show optimization with the additional penalty in \eqref{penalty} and the right hand side plots show optimization using a truncated parameters search space where the optimizer was restricted to search $\gamma$ and $-\beta$ up to $\pi$ and $\pi/2$ as in \eqref{truncation}.}
    \label{fig:All_graphs}
\end{figure}

\clearpage
\end{document}